\documentclass[11pt]{article}

\usepackage[utf8]{inputenc} 
\usepackage[T1]{fontenc}    
\usepackage{microtype}      

\usepackage{amsmath}
\usepackage{colortbl}
\usepackage{cleveref}
\usepackage{booktabs}
\usepackage{dsfont}
\usepackage{algorithm}
\usepackage{algorithmic}
\usepackage[dvipsnames]{xcolor}
\usepackage{multicol}
\usepackage{multirow}
\usepackage{ctable}
\usepackage{amssymb}
\usepackage{amsfonts}       

\usepackage{mathtools}
\usepackage{amsthm}
\usepackage{subcaption}

\providecommand{\norm}[1]{\ensuremath{\left\lVert#1\right\rVert}}
\providecommand{\inner}[1]{\ensuremath{\left\langle#1\right\rangle}}

\newcommand{\eps}{\epsilon}
\newcommand{\ind}{\mathds{1}}
\DeclareMathOperator*{\E}{\mathbb{E}}
\newtheorem{theorem}{Theorem}[section]
\newtheorem*{theorem*}{Theorem}
\newtheorem{proposition}[theorem]{Proposition}
\newtheorem{lemma}[theorem]{Lemma}
\newtheorem*{lemma*}{Lemma}

\newtheorem{corollary}[theorem]{Corollary}
\newtheorem{definition}[theorem]{Definition}

\newtheorem{remark}[theorem]{Remark}

\DeclareMathOperator*{\argmax}{arg\,max}

\newcommand{\Expected}[2]{\mathbb{E}_{#1}\left[ #2 \right]}

\DeclareMathOperator{\MMD}{MMD}

\newcommand{\calD}{\ensuremath{\mathcal{D}}}

\newcommand{\calF}{\ensuremath{\mathcal{F}}}
\newcommand{\calG}{\ensuremath{\mathcal{G}}}

\newcommand{\calM}{\ensuremath{\mathcal{M}}}

\newcommand{\calX}{\ensuremath{\mathcal{X}}}

\newcommand{\calZ}{\ensuremath{\mathcal{Z}}}

\newcommand{\bX}{\ensuremath{\mathbf{X}}}

\newcommand{\alglinelabel}{%
  \addtocounter{ALC@line}{-1}
  \refstepcounter{ALC@line}
  \label
}

\newcommand{\newlib}{\textsc{DP-Auditorium }}
\newcommand{\renyitester}{\textsc{R\'enyiTester }}
\newcommand{\gilberttester}{\textsc{HistogramTester }}
\newcommand{\hstester}{\textsc{HockeyStickTester }}
\newcommand{\mmdtester}{\textsc{MMDTester }}

\newcommand{\dpmean}{\textsc{DPLaplace }}

\newcommand{\nondplaplaceone}{\textsc{NonDPLaplace1 }}
\newcommand{\nondplaplacetwo}{\textsc{NonDPLaplace2 }}
\newcommand{\nondpgaussianone}{\textsc{NonDPGaussian1 }}
\newcommand{\nondpgaussiantwo}{\textsc{NonDPGaussian2 }}

\newcommand{\svtone}{\textsc{SVT1 }}
\newcommand{\svttwo}{\textsc{SVT2 }}
\newcommand{\svtthree}{\textsc{SVT3 }}
\newcommand{\svtfour}{\textsc{SVT4 }}
\newcommand{\svtfive}{\textsc{SVT5 }}
\newcommand{\svtsix}{\textsc{SVT6 }}
\newcommand{\joint}{\textsc{JointMedian }}
\newcommand{\noisymax}{\textsc{NoisyMax }}

\newcommand{\scaledSGD}{\textsc{ScaledGD }}

\newcommand{\arxiv}[1]{#1}
\newcommand{\narxiv}[1]{}

\begin{document}

\title{DP-Auditorium: a Large Scale Library for Auditing Differential Privacy}

\author{William Kong, Andr\'es Mu\~noz Medina,\\ M\'onica Ribero, and Umar Syed \thanks{Google Research, NY, \{weiweikong, ammedina, mribero, usyed\}@google.com}
}

\maketitle

\begin{abstract}
New regulations and increased awareness of data privacy have led to the deployment of new and more efficient differentially private mechanisms across public institutions and industries. Ensuring the correctness of these mechanisms is therefore crucial to ensure the proper protection of data. However, since differential privacy is a property of the mechanism itself, and not of an individual output, testing whether a mechanism is differentially private is not a trivial task. While ad hoc testing techniques exist under specific assumptions, no concerted effort has been made by the research community to develop a flexible and extendable tool for testing differentially private mechanisms. This paper introduces \newlib as a step advancing research in this direction. \newlib abstracts the problem of testing differential privacy into two steps: (1) measuring the distance between distributions, and (2) finding neighboring datasets where a mechanism generates output distributions maximizing such distance. From a technical point of view, we propose three new algorithms for evaluating the distance between distributions. While these algorithms are well-established in the statistics community, we provide new estimation guarantees that exploit the fact that we are only interested in verifying whether a mechanism is differentially private, and not in obtaining an exact estimate of the distance between two distributions. \newlib is easily extensible, as demonstrated in this paper by implementing a well-known approximate differential privacy testing algorithm into our library. We provide an extensive comparison to date of multiple testers across varying sample sizes and differential privacy parameters, demonstrating that there is no single tester that dominates all others, and that a combination of different techniques is required to ensure proper testing of mechanisms.  
\end{abstract}


\section{Introduction}
\label{sec:intro}
The protection of data contributors' privacy is a growing concern as data-driven algorithms become increasingly ubiquitous. Differential privacy (DP) addresses this concern by developing mechanisms, which are private counterparts for traditional algorithms that consume user data and that guarantee with high probability that the output of the mechanism does not reveal any information about a user that could not have been known without the user being part of the input data. Specifically, a mechanism is differentially private if, for neighboring datasets (i.e., datasets that differ in only one record), an adversary is unable to determine from the output of the mechanism which of the neighboring datasets the output came from. Differential privacy is widely used by industry and governments as it meets strict user rights and data protective regulations such as the General Data Protection Regulation (GDPR) and the Digital Markets Act (DMA).

 As DP becomes more prevalent, so too do the risks of incorrect implementations, which can potentially jeopardize user privacy. Accordingly, auditing differentially private mechanisms is a critical step for increasing confidence in the correctness of DP pipelines and ensuring that user rights are met. However, the variety of mechanisms, tasks, and generalized definitions for differential privacy make it difficult to inspect the correctness of this vast amount of algorithms. This suggests that automated and generalizable end-to-end tests for the implementation of such mechanisms are desirable for expanding the confidence of their correctness.

In this paper, we take a step forward in standardizing the way in which differential privacy is audited. We introduce \newlib, a DP auditing library implemented in Python that allows practitioners to test DP guarantees from only black-box access to a mechanism and allows for continuous improvements from the research community. \newlib also introduces a new family of function-based tests for R\'enyi, pure, and approximate differential privacy. This family allows for black-box detection of a variety of bugs without knowledge or assumptions about the source of randomness (e.g., noise decay, or variance).

\newlib implements two main components: testers and dataset finders. Testers take a mechanism and a pair of datasets as input and aim to find violations of a privacy guarantee on the provided datasets. Dataset finders suggest the neighboring datasets. \newlib  combines these components allowing for flexible testing of diverse mechanisms and privacy definitions in an automated way and is able to catch bugs reported in the literature as well as implementation bugs in our own experiments for this paper (see Section~\ref{sec:svt4}).

\subsection{Contributions.}
\begin{itemize}
    \item 
    We introduce a family of function-based testers for R\'{e}nyi, pure, and approximate differential privacy. These testers use variational representations of divergence measures between distributions well known to the statistical community to lower bound the true privacy parameters of a mechanism on a given pair of datasets. The specific application to privacy allows us to derive new finite complexity bounds for estimating the Hockey Stick divergence, the R\'{e}nyi divergence, and the maximum mean discrepancy (MMD). While the complexity of previous histogram-based testers depends on the discretization of the output space \cite{GM18, BSBV21}, our testers optimize over a family of functions that parametrize the complexity of the test.  We also provide an industry ready implementation that consistently outperforms histogram based divergence estimators.  This new function based approach is also more flexible and is amenable to testing high-dimensional and continuous output spaces, without an explicit discretization of the output space.
    \item We implement dataset finders that suggest neighboring datasets. We implement known approaches for finding neighboring datasets, namely grid search and randomized search. We expand this suite by incorporating Bayesian optimization techniques that generate datasets maximizing the divergences we wish to estimate. While grid search is advantageous when side-information about potential bugs is available, optimization-based finders accelerate and in some cases are even the determinant in finding certain bugs.  
    \item We conclude our paper with an extensive comparison of testers and dataset finders. We show that they are able to detect a variety of bugs. While experiments show that no tester is globally better, we provide insight into what testers are more suitable for detecting certain bugs under different privacy definitions.
    \item \newlib introduces a \textit{tester} wrapper that can integrate most if not all of the previous black-box testing work. The library will be open sourced and open to researchers and practitioners to experiment and test DP mechanisms. 
\end{itemize}

The remainder of the paper is structured as follows. In section \cref{sec:setup}, we provide preliminary notions and notation. In \cref{sec:testers}, we formalize privacy tests and dataset finding algorithms. We conclude our paper in \cref{sec:experiments} by evaluating and discussing different testers. 

\subsection{Related work. }\label{sec:related_work}

There are generally two approaches used in differential privacy testing that we review below.

\textbf{Property testing literature.} The first approach, which includes our proposed methodology, aims to directly estimate the \textit{effective} privacy parameters from black-box access to the tested mechanism and compare these effective privacy parameters with the ones stated by the privacy guarantee. This approach focuses on estimating the distance between the distributions induced by the mechanism on neighboring datasets. The problem thus consists of two key challenges: 1) estimating the distance between distributions given two fixed neighboring datasets, and 2) finding the pair of neighboring datasets that maximize the distance between these distributions.

The problem of estimating distances between distributions has been extensively studied in the statistics and hypothesis testing community. While providing a comprehensive overview of this literature is beyond the scope of this work, we briefly highlight \cite{NWJ10,zhao,sriperumbudur2012empirical}, which consider estimating probability distances through optimization methods over function spaces. Their work provides asymptotic guarantees, while we provide strict finite sample complexities to obtain a lower bound on divergences between two distributions.

For the specific task of estimating R\'enyi divergence, our estimator is inspired by the work of \cite{BDKRW21}, which considers using neural networks to estimate R\'enyi divergence. The finite sample complexity bounds provided in that work, however, depend on the structure of the neural network and can rapidly become vacuous for the purpose of testing differential privacy. In contrast, our complexity bounds are independent of the network structure, as we are primarily concerned with lower bounds on the R\'enyi divergence. 

Several related approaches lower bound privacy parameters but make assumptions
on the type of mechanisms to provide finite sample complexity guarantees.
For example, \cite{DM22} proposes to estimate the \textit{regularized kernel R\'enyi
divergence}, a lower bound on the R\'enyi divergence between distributions of a
randomized mechanism. However, this approach requires knowledge about the
covariance matrix of the underlying distributions, which is impractical for most
mechanisms other than the Gaussian and Laplace mechanisms, and inaccessible in
the black-box setting. Similarly, \cite{DGKK22} provides techniques for
lower-bounding $\epsilon$. Unfortunately, the previous method needs access to the cumulative distribution function of the distribution of the privacy loss random variable, which is unknown under our assumption of black box mechanism evaluation. None of the aforementioned techniques considers the second problem of finding datasets where the privacy guarantee can break.

 A large body of literature exists regarding the testing of a mechanism's privacy for specific settings or under assumptions on the mechanism. \cite{DJT13} proposes a differential privacy tester for mechanisms with discrete and finite output spaces that requires access to the distribution over datasets and the probability measure over outputs induced by the tested mechanism. For the dataset search, instead of testing privacy in the worst case setting, they test if the mechanism satisfies the guarantee over fixed datasets with high probability. More importantly, the tester does not work for continuous output spaces. \cite{GM18} presents a test for discrete $(\epsilon, \delta)$-DP mechanisms but omits the problem of finding the worst case pair of neighboring datasets. StatDP \cite{statdp} proposes a system for detecting differential privacy violations by post-processing the output of the mechanisms through different statistics. The tester requires semi-black box access to the mechanisms (as one of the post processing techniques requires running the mechanism without privacy), which is infeasible for auditing certain mechanisms.

DP-Sniper \cite{dpsniper} provides an $\epsilon$-DP black-box mechanism tester that explicitly attempts to find a set in the output space that maximizes the difference in probability for the output of the mechanism. Their framework is specific to detecting $\epsilon$-DP, and the high sample complexity (millions of samples) makes this approach impractical for some mechanisms. On the dataset search problem, DPSniper uses predetermined rules that may hinder its ability to detect violations on new tasks, and it runs under non-classic neighboring relations such as the $\ell_{\infty}$ relation instead of the classic swap or add/remove definition of neighboring.

\textbf{ML focused testing literature.} The second approach focuses on machine learning predictive models learning algorithms rather than arbitrary statistical tasks \cite{JUO20}.

On a different vein, some work uses adversarial attacks that try to break the privacy definition. Examples of these approaches are the celebrated membership inference attacks \cite{JE19,ROF21,CYZF20} and data reconstruction attacks \cite{Guo22, Balle22} of deep learning models trained with DP-SGD. Hence, the validation of whether a mechanism satisfies privacy is linked to the ability of the attack to succeed. The tests generated by these approaches are very valuable when trying to understand potential privacy risks on a single data set, by manually designing canaries that are expected to have highest sensitivity. For example,  \cite{LMFZ22} extends the work of \cite{dpsniper} by developing data poisoning attacks to explore the space of datasets focusing on machine learning predictive models learning algorithms rather than arbitrary statistical tasks.  

Running these tests generally requires white box access to the trained model and, more importantly, requires access to large portions of the training data, making auditing of a privately trained model impossible for someone who is not the data curator. 
Consequently, the resulting lower bounds from these approaches tend to be loose \cite{NSTP+, NHSB+23}. Moreover, the budget $\epsilon$ predicted by these experiments is generally much smaller than the theoretical budget. For example, some authors \cite{jagielski2020auditing} assert that their proposed models were private with an $\epsilon = 10^{-3}$ when these models were trained without privacy. 

Dataset search in this space typically relies on pre-determined datasets, or canaries, which do not attempt to understand the worst-case (unknown) scenario that differential privacy aims to protect. However, recent work has introduced a randomized approach to crafting canaries that, with high probability, finds datasets that break the privacy guarantee for the Gaussian mechanism in high dimensions \cite{AKOOMS23}.

\section{Preliminaries}
\label{sec:setup}

\textbf{Notation.} For $n$ a natural number, $[n]=\{1, ...,n\}$ denotes the set the set of positive numbers up to $n$. 
We reserve $\epsilon, \alpha, \delta$ for privacy parameters, $\gamma$ will represent the error tolerance, $\beta$ will represent the probability of failure of a given test. 

We start by formally defining the tested privacy properties and notions.

\begin{definition}
\label{def:neighbor}
Let $D_0, D_1 \subset \calD$ be datasets, where $\calD$ is a domain of datasets that consist of records taking values in  a domain $\calZ$. We say that $D_0$ and $D_1$ are neighbors, denoted $D_0 \sim D_1$, if $D_0$ is obtained by adding or removing a record $z\in\calZ$ from $D_1$. 
\end{definition}

\begin{remark}
The notion of neighboring can be replaced by the swap definition where $D_0$ is obtained by swapping one record in $D_1$ by another record.
\end{remark}

A \emph{mechanism} is a randomized function $\calM: \calD \rightarrow \calX$ (where appropriate, we denote
$\calM(D)$ as either a single output of mechanism $\calM$ on an input $D$ or the distribution of the outputs). A mechanism is differentially private if it generates `close' distributions for neighboring datasets.  Several measures of closeness have been proposed. 

\begin{definition}
\label{def:dp}
A randomized mechanism $\calM: \calD \to \calX$ satisfies $(\epsilon, \delta)$--approximate differential privacy ($(\epsilon,\delta)$--DP), if for every pair of neighboring datasets $D_0$ and $D_1$ and every subset $A \subseteq \calX$ of the output space, it holds that
\begin{equation}
\label{eq:dp}
    P(\calM(D_0) \in A) \leq e^{\epsilon}P(\calM(D_1) \in A) +\delta.
\end{equation}
$\calM$ satisfies pure differential privacy, or is $\epsilon$--differentially private ($\epsilon$--DP), when $\delta=0$.
\end{definition}



A different and frequently used divergence in the differential privacy literature is the R\'enyi divergence.
\begin{definition}
\label{def:renyi-divergence}
Let $P$ and $Q$ be two distributions. The R\'enyi divergence of order $\alpha > 1$ between the two distributions is given by
\begin{equation*}
    R_\alpha(P ||Q) := \frac{1}{\alpha - 1} \log \Expected{Q}{\left(\frac{dP}{dQ}\right)^\alpha}.
\end{equation*}
\end{definition}

\begin{definition}
\label{def:renyi-dp}
A randomized mechanism $\calM: \calD \to \calX$ satisfies $(\alpha, \epsilon)$--R\'enyi differential privacy if for any pair of neighboring datasets $D_0$ and $D_1$, it holds that
 \begin{equation}
     \label{eq:def-renyi-dp}
     R_{\alpha}(\calM(D_0)||\calM(D_1)) \leq \epsilon
 \end{equation}
 \end{definition}

\section{Testing Framework}
\label{sec:building-blocks}

This section presents the main DP testing framework of our library and discusses some of its properties.

As seen in the previous section, asserting that a mechanism is private requires verifying that the  output distributions $\calM(D_0)$ and $\calM(D_1)$  are close under some divergence for \emph{all} neighboring datasets. If there exists a pair of neighboring datasets $D_0$ and $D_1$ for which the divergence of $\calM(D_0)$ and $\calM(D_1)$ is large, it follows that the mechanism is not private. Since we treat mechanism $\calM$ as a black-box, we cannot observe the true output distributions of the mechanism and exact calculation of the divergence between $\calM(D_0)$ and $\calM(D_1)$ is not possible. Instead, we introduce Algorithm~\ref{alg:general} that suggests a general procedure for finding violations of privacy properties of a mechanism, relying on samples from the mechanism to estimate the divergence. 

It is worth mentioning that our library focuses in completeness but not soundness of the tests. This means that a mechanism failing the test will reject the hypothesis of a mechanism's privacy guarantee, but passing the test is not enough evidence to accept the hypothesis of the mechanism's privacy guarantee. It has been shown \cite{GM18} that sound tests for pure DP have exponential sample complexity. Sound tests for approximate DP have been developed only for finite output space mechanisms, and over a fixed pair of datasets. Testers for black-box mechanisms that can deal with exponentially large output spaces, or that can explore the space of datasets  and confirm the privacy guarantee of a mechanism is in general a hard open problem.

\newlib abstracts the two main steps of \cref{alg:general} and allows users to specify any divergence estimator and neighboring generating routines. The remainder of the paper will focus on specific implementations of the divergence estimator $\widehat{D}(\cdot\|\cdot)$ as well as the introduction of non-parametric Bayesian optimizers to handle the process of generating neighboring datasets.


\begin{algorithm}
\caption{Generalized lower bound test method for DP}
\label{alg:general}
\begin{algorithmic}[1]
\STATE{{\bf Inputs:} mechanism $\calM$, trial count $T$, a dataset generator $\mathcal G$, divergence estimator $\widehat{D}$, privacy threshold $\tau$, probability of failure $\beta$;}
\FOR{$t = 1, \ldots T$}
\STATE{Generate datasets $\{D_0^t, D_1^t\}$ using $\mathcal G$.}  
\STATE{Set ${\calD}_1^t := \widehat{D}(\calM(D_0^t)\,\|\,\calM(D_1^t))$.}
\STATE{Set ${\calD}_2^t := \widehat{D}(\calM(D_1^t)\,\|\,\calM(D_0^t))$.}
\alglinelabel{algln:threshold} \IF{$\max\{{\calD}_1^t, {\calD}_2^t\}  >  \tau$}
\RETURN{w.p. $1-\beta$, $\calM$ is not private}
\ENDIF
\ENDFOR
\RETURN{Could not find privacy violation in $T$ trials.}
\end{algorithmic}
\end{algorithm}


\section{Implementations of the Framework}
\label{sec:testers}
This section presents implementations of the framework given in Algorithm~\ref{alg:general}. More specifically, it considers the following two types:

\begin{itemize}
\item \textit{Histograms}. These methods estimate a particular divergence using histograms (e.g., see \cite{GM18}), which can (i) approximate distribution to arbitrary precision and (ii) provide examples of where privacy violations may occur. However, these methods are prohibitively expensive in high-dimensional settings and large (mechanism) output spaces. \vspace*{1em}


\item \textit{Dual divergences}. These methods lower bound a particular divergence using its dual formulation. Specifically, these methods solve the associated empirical maximization problem using some well-known ML heuristics that scale well to high-dimensional settings. 
\end{itemize}

It is worth noting that the dual formulation mentioned above is well-known to the statistics community. However, to the best of our knowledge, we provide the first practical implementation for the purpose of testing differential privacy. In Section~\ref{sec:experiments}, we also show that dual divergence-based approaches generally outperform histogram-based methods in practice.


\subsection{Pure and R\'enyi DP}

This section presents a tester that checks if a mechanism satisfies a given level of pure or R\'enyi DP. It primarily focuses on developing results for R\'enyi DP and the tester implements the framework in Algorithm~1.

\subsubsection{R\'enyi divergence lower bound}
We first introduce a variational formulation of the R\'enyi divergence based on the results in \cite{NWJ10, BDKRW21}. 


\begin{theorem}[Theorem 3.1 in \cite{BDKRW21}]
\label{thm:variational-renyi}
 Let $P$ and $Q$ be probability measures on $(\calX, \calF)$ and $\alpha>1$. 
Let $\Gamma$ be any function space such that $M_b(\calX) \subseteq \Gamma \subseteq M(\calX)$ where $M_b(\calX)$  and $M(\calX)$ are the sets of measurable bounded and measurable functions on $\calX$ respectively. Then,
\begin{align}
\nonumber & R_{\alpha}(P\,\|\,Q) =  \sup_{h \in \Gamma}  \tfrac{\alpha}{\alpha-1}\log \left(\Expected{X\sim P}{ e^{(\alpha-1)h(X)}}\right)\\
\label{eq:variational-renyi-formula} & \hspace{2.5cm} - \log\left(\Expected{X\sim Q}{ e^{\alpha h(X)}}\right)
\end{align}
Moreover, if $P$ and $Q$ admit densities, then the function optimizing the above quantity is given by $h = \log P/Q.$
\end{theorem}
Since the complexity of the function space $\Gamma$ can be arbitrarily large, computing the supremum in \cref{eq:variational-renyi-formula} is generally intractable. Hence, instead of estimating $R_\alpha(\cdot\|\cdot)$ to an arbitrary precision, we propose to (i) replace $P$ and $Q$ with samples $\bX_{0}$ and $\bX_{1}$, respectively, and (ii) replace $\Gamma$ with a (possibly) more restrictive function space $\Phi$. Under these modifications, we then use a one-sided confidence interval to obtain a suitable lower bound for $R_\alpha(\cdot\|\cdot)$. 

Let us now make these relaxations formal through the following definition.



\begin{definition}
\label{def:renyi-partial-defs}
Let $h:\calX \to \mathbb{R}$ be a function in $\Phi$ and $\alpha>1$. Define
\begin{align}
  \nonumber  R^h_{\alpha}(P\,\|\,Q) &:= \tfrac{\alpha}{\alpha-1} \log \left( \int e^{(\alpha-1) h(x)}dP\right) \\
  & \hspace{1.cm} -  \log \left( \int e^{\alpha h(x)}dQ\right)
\end{align}
 Given samples $\bX_0 = (X_{0,1},..., X_{0,n}) \sim P^n$, $\bX_1 = (X_{1,1},...,X_{1,n}) \sim Q^n$, define its empirical counterpart
 \begin{align}
   \nonumber R^{h,n}_{\alpha}(\bX_0\,\|\,\bX_1) &:= \tfrac{\alpha}{\alpha-1} \log \left( \frac{1}{n} \sum_{i=1}^n e^{(\alpha-1) h(X_{0,i})}\right)\\
    \label{eq:empirical-renyi} &\hspace{1.cm} -  \log \left( \frac{1}{n} \sum_{i=1}^n e^{ \alpha h(X_{1,i})}\right).
 \end{align}
Let $h\in \Phi \subseteq \Gamma$, define the (one-sided) estimation error as the smallest non negative $\gamma(\Phi, n, \beta, \alpha) >0$ such that 
\begin{equation}
\label{eq:threshold_cond}
R^{h}_\alpha(P\,\|\,Q) > R^{h,n}_\alpha(\bX_{0}\,\|\,\bX_{1})-\gamma(\Phi, n, \beta, \alpha)
\end{equation}
with probability at least $1-\beta$ over the randomness of the $n$ samples from $P$ and $Q$.
\end{definition}

Observe that the parameter $\gamma$ can depend on the complexity of the function class. As a concrete example, we characterize the complexity of the function class 
\begin{equation} 
\Phi_C := \{h \in \Gamma : |h(x)|<C \text{ for all } x \in \calX \} \label{eq:PhiC}
\end{equation}
and prove a suitable relationship between $\gamma$ and $n$.

\begin{theorem}
\label{thm:sample_complexity_renyi}
Let $P$ and $Q$ be two distributions. Let $\Gamma$ be as in Theorem~\ref{thm:variational-renyi}, $\Phi_C \subseteq \Gamma$ be as in \eqref{eq:PhiC}, and $h \in \Phi_C$. Let $\bX_0 = (X_{0,1},...,X_{0,n})$ and $\bX_1 = (X_{1,1},...,X_{1,n})$ be $n$  realizations of $P$ and $Q$, respectively. Then, if $\eta \in (0,1]$, and 
\[
n \geq \max\left(\frac{3e^{2(\alpha-1)C}\log(2/\beta)}{\eta^2}, \frac{2e^{\alpha C}\log(2/\beta)}{\eta^2} \right) =: \underline{n},
\]
then, with probability at least  $1-\beta$, it follows that 
\begin{align}
\label{eq:renyi_lower}
 R_\alpha&(P||Q) \geq
 R^{h,n}_{\alpha}(\bX_0|| \bX_1)  - \log \left(\frac{1+\eta}{1-\eta} \right).
\end{align}
\end{theorem}


We now make two remarks about the above result. First, the larger $C$ is, the larger $\underline{n}$ and $\Phi_C$ are. Consequently, the tightness of the lower bound in \eqref{eq:renyi_lower} can be improved by increasing $C$ at the cost of an exponential increase in the sample complexity given by $\underline{n}$. Second, in contrast to well-known exponential sample complexities \cite{KKPW14} for estimating $R_\alpha(\cdot\|\cdot)$, the sample complexity in Theorem~\ref{thm:sample_complexity_renyi} is dimension-independent. This is
 primarily due to the fact that $\Phi_C$ is a relatively much smaller space compared to $\Gamma$ in the context of \eqref{eq:variational-renyi-formula} and the fact that we are only interested in obtaining a lower bound on $R_\alpha(P||Q)$. Finally, it is important to notice that, even though there is an exponential dependency on $C$, since we know the optimal solution is given by $\log(P/Q)$ we can generally set $C$ to be a small constant and still ensure a tight lower bound on the approximation error to the R\'enyi divergence.
 


To conclude, we show how the above theorem can be used to characterize the log-error in \eqref{eq:renyi_lower} in terms of $n$.

\begin{corollary}
\label{cor:renyi-error-bound}
Let $n \geq 1$ be fixed. With probability $1-\beta$, the bound in \eqref{eq:renyi_lower} holds with \begin{equation}
    \label{eq:eta_n}
    \eta = \sqrt{\frac{\max\left\{3e^{2(\alpha-1)C}, {2e^{\alpha C}} \right\} \cdot \log(2/\beta) }{n}}.
\end{equation}
\end{corollary}



\subsubsection{R\'enyi tester}
\label{sec:function-spaces}
We now present how to use the previous results to implement a R\'enyi divergence-based tester under the framework in Algorithm~\ref{alg:general}.

Recall that the previous section showed that we can choose an arbitrary function $h \in \Phi$ to lower bound the R\'enyi divergence between the output of a mechanism in two neighboring datasets. Since we would like the tightest bound for \eqref{eq:renyi_lower}, a natural heuristic for choosing $h$ is to first sample ${\bX_0 = (X_{0,1}, \ldots, X_{0,n})}$ from $\calM(D_0)$, ${\bX_1 = (X_{1,1}, \ldots, X_{1,n})}$ from $\calM(D_1)$, and then choose $h$ to be the maximum of $R^{h,n}_\alpha(\cdot\|\cdot)$ over $\Phi_C$ (or some subspace $\Psi_C\subseteq\Phi_C$). 

Using the above choice of $h$, and appropriate choices of thresholds and sampling errors, we now present the complete R\'enyi divergence-based tester in Algorithm~\ref{alg:renyi-tester} which implements Algorithm~\ref{alg:general}.

\begin{algorithm}[tb]
  \caption{\renyitester}
  \label{alg:renyi-tester}
    \begin{algorithmic}[1]
    \STATE {\bfseries Inputs:}   \alglinelabel{algln:input} Parameters $\calM$, $T$, $\mathcal G$, and $\beta$ as in Algorithm~\ref{alg:renyi-tester}; $\Psi_C\subseteq\Phi_C$ where $\Phi_C$ is as in \eqref{eq:PhiC}; R\'enyi DP parameters $(\alpha, \epsilon)$; and sample size $n$;
   \STATE  Compute $\eta$ as in \eqref{eq:eta_n} and \[\tau:= \begin{cases}\epsilon & \text{if testing R\'enyi DP,}\\
      \min (\epsilon, 2\alpha\epsilon^2) & \text{if testing pure DP.}
      \end{cases}
    \] 
    \STATE Define the estimator
    \[
    \widehat{D}(P\,\|\,Q) := \max_{h\in\Psi_C} R^{h,n}_\alpha(\textbf{X}_0\,\|\,\textbf{X}_1) - \log\left(\frac{1+\eta}{1-\eta}\right)
    \]
    for any $(P,Q)$, where $\textbf{X}_{0} \sim P^n$ and $\textbf{X}_{1} \sim Q^n$.
    \STATE{Run the same steps as in Algorithm~\ref{alg:general} with inputs $\calM$, $T$, $\mathcal G$, $\widehat{D}$, $\tau$, and $\beta$.}
    \end{algorithmic}
\end{algorithm}

Some remarks about Algorithm~\ref{alg:renyi-tester} are in order. First, it is an instance of Algorithm~\ref{alg:general} which sets specific values for $(\tau, \widehat{D})$, while keep the other parameters in Algorithm~\ref{alg:general} as inputs. Second it is used to evaluate only pure DP or R\'enyi DP of a mechanism $\calM$. Third, its efficiency is generally to the efficiency of computing $\max_{h\in\Psi_C} R^{h,n}_\alpha(\cdot\|\cdot)$. 

In our numerical experiments, we make two concessions for the sake of tractability. First, we consider two specific choices of $\Psi_C$, namely, (i) the family of dense neural networks with bounded output and (ii) the vector space of polynomials generated by the basis of Chebyshev polynomials of degree at most $k$. Note that these particular choices of $\Psi_C$ make the optimization problem 
$\max_{h\in\Psi_C}R^{h,n}_\alpha(\cdot\|\cdot)$ tractable when applying first-order method such as stochastic gradient descent (SGD). Second, we estimate $\max_{h\in\Psi_C} R^{h,n}_\alpha(\textbf{X}_0\|\textbf{X}_1)$, by (i) independently sampling 
$\textbf{X}_0, \tilde{\textbf{X}}_0 \sim P^n$ 
and $\textbf{X}_1, \tilde{\textbf{X}}_1 \sim Q^n$, 
(ii) computing $h^* := \argmax_{h\in\Psi_C} R^{h,n}_\alpha(\textbf{X}_0\|\textbf{X}_1)$; and 
(iii) evaluating $R^{h^*,n}_\alpha(\tilde{\textbf{X}}_0\|\tilde{\textbf{X}}_1)$ as our estimate.

\subsubsection{Adaptation of R\'enyiTester to pure DP}

It is known that setting $\alpha=\infty$ for R\'enyi DP corresponds to pure-DP, i.e., $\calM$ is an $\epsilon$--DP mechanism if and only if for any ${D\sim D'}$, we have $D_\infty(\calM(D) || \calM(D')) \leq \epsilon$.  Lemma~\ref{lemma:renyi-divergence-monotone} below states the R\'enyi divergence is monotone in $\alpha$. This implies that we can use our tester with any $\alpha>1$ to test for pure DP since ${D_\infty(\calM(D) || \calM(D')) \geq D_\alpha(\calM(D) || \calM(D'))}$, so any lower bound for $D_\alpha$ will also be a lower bound for $D_\infty$.

\begin{lemma}[Proposition 9 in \cite{M17}]
\label{lemma:renyi-divergence-monotone}
Let $1<\alpha_1<\alpha_2$ and $P$ and $Q$ be probability measures. Then $D_{\alpha_1}(P||Q) < D_{\alpha_2}(P||Q)$
\end{lemma}

In the specific case of testing pure DP, Lemma~\ref{lem:renyi_suff_cond} gives a bound for the divergence that can be used to obtain a tighter threshold for the test (see line~\ref{algln:threshold}).  

\begin{lemma}[Lemma 1 in \cite{M17}]
\label{lem:renyi_suff_cond}
Let $\calM$ be an $\epsilon$--differentially private mechanism and $\alpha>1$. Then $D_{\alpha}(\calM(D)||\calM(D'))\leq \min\{\epsilon, 2 \alpha \epsilon^2\}$. 
\end{lemma}



\subsection{Approximate DP}

This section presents testers that check if a mechanism satisfies a given level of approximate DP. Similar to the previous section, each tester implements the framework in Algorithm~\ref{alg:general}.
\subsubsection{Histogram Tester}
The authors in paper \cite{GM18} introduce a test for approximate differential privacy over discrete and finite output spaces. We introduce this through Algorithm~\ref{alg:gilbert_divergence} and Algorithm~\ref{alg:gilbert_test}. It requires knowledge of the output space. The test estimates the parameter $\delta$ by directly computing the privacy loss over different outputs.

\begin{algorithm}
  \caption{Estimate histogram divergence}
  \label{alg:gilbert_divergence}
\begin{algorithmic}[1]
  \STATE {\bfseries Input:} Distributions $P$ and $Q$; universe $[m]$; $\epsilon>0$; $\beta>0$; and $\eta > 0$;
  \STATE Compute $\lambda = \max \bigl\{4m(1+e^{2\epsilon})\beta^{-2}, 12(1+e^{2\epsilon})\eta^{-2} \bigr\} $
  \STATE Sample $r \sim \text{Poisson}(\lambda)$
  \STATE Generate samples $X \sim P^r$ and $Y \sim Q^r$
  \FOR{$j=1$ {\bfseries to} $m$}
        \STATE $x_j= |\{i : X_i = j, 1\leq i\leq r\}|$
        \STATE $y_j=|\{i : Y_i = j, 1\leq i\leq r\}|$
    \alglinelabel{algln:adp_probability} \STATE $z_i = (x_i-e^{\epsilon}y_i) / r$
  \ENDFOR
  \RETURN $-\eta + \sum^r_{i=1} \max\{0,z_i\}$
\end{algorithmic}
\end{algorithm}

\begin{algorithm}
  \caption{\gilberttester}
  \label{alg:gilbert_test}
\begin{algorithmic}[1]
  \STATE {\bfseries Input:} Parameters $\calM$, $T$, $\mathcal G$, and $\beta$ as in Algorithm~\ref{alg:general}; parameters $\eta$ as in Algorithm~\ref{alg:gilbert_divergence}; approximate DP parameters $(\epsilon,\delta)$; and sample size $m$;
  \STATE Define the estimator $\widehat{D}(P\,\|\,Q)$ to be the output of Algorithm~\ref{alg:gilbert_divergence} with inputs $P$, $Q$, $[m]$, $\epsilon$, $\beta$, and $\eta$.
    \STATE Set $\tau = \delta$.
    \STATE{Run the same steps as in Algorithm~\ref{alg:general} with inputs $\calM$, $T$, $\mathcal G$, $\widehat{D}$, $\tau$, and $\beta$.}
\end{algorithmic}
\end{algorithm}

\subsubsection{Hockey-stick divergence}
\label{sec:hs_div}
It is well known that DP can be defined in terms of the Hockey stick divergence (see Lemma 5 in \cite{ZDW22}).
\begin{lemma}{(\cite{BO13})}
Consider the function $f(x) = \max\{0, x-e^{\epsilon} \}$. The Hockey-Stick divergence between $P$ and $Q$ (of order $e^{\epsilon}$) is the $f$-divergence defined as
\begin{equation}
    \label{eq:hockey-stick}
    H_{\epsilon}(P||Q):= \Expected{Q}{f\left(\frac{dP}{dQ} \right)}.
\end{equation}

A mechanism is $(\epsilon, \delta)$--DP if and only if 
\begin{equation*}
\label{eq:dp-as-hockey-stick}
    \sup_{D_0\sim D_1}H_\epsilon(\mathcal{M}(D_0)||\mathcal{M}(D_1)) \leq \delta
\end{equation*}
where the supremum is taken over neighboring datasets.
\end{lemma}

\begin{lemma}
The variational representation for the Hockey stick divergence (see \cref{eq:var-f-divergence} in the supplementary material) is given by

\begin{equation}
\label{eq:hs-variational}
H_{\epsilon}(P||Q) = \sup_{g:0\leq g\leq 1}\Expected{P}{g(X)}-\Expected{Q}{e^{\epsilon}g(X))}.
\end{equation}
\end{lemma}

Similar to the R\'enyi tester, the Hockey-stick tester is based on the dual formulation given in \eqref{eq:hockey-stick}. In this section, we further manipulate the expression for the Hockey-stick divergence and establish a connection between that divergence and an unbalanced binary classification problem. 
\begin{lemma}
 \label{lem:classification} 
 Let $\epsilon > 0$ Fix two distributions $P$ and $Q$ over $\calX$. Let $\mathbb{D}$ be a mixture distribution over $\calX \times {0, 1}$ defined as $\mathbb{D} = \frac{e^\epsilon}{1 + e^\epsilon} Q \times \delta_0 + \frac{1}{1 + e^\epsilon} P \times \delta_1$. That is, $(X, Y) \sim \mathbb{D}$ is a weighted, labeled sample from $Q$ (with label $0$) and  $P$ (with label $1$).
 $$H_\epsilon(P || Q) = (1 + e^\epsilon) \sup_{g \colon \calX \to \{0,1\}} \mathbb{D}(g(X) = Y) - e^\epsilon$$
 \end{lemma}

The advantage of this formulation is that it allows for a simple derivation on how well we can approximate the true divergence with a sample.
\begin{definition}
\label{eq:lb_est}
Let $S_m$ denote a binomial random variable of parameters $m, p$. Let $\gamma_m, \beta > 0$ We say that $\widehat p \colon= \widehat p (S_n)$ is a $(\gamma_m, \beta)$ lower bound estimator for $p$ if with probability at least $1 - \beta$ over the draw of $S_n$ we have that $p > \widehat p(S_m) - \gamma_m$.
\end{definition}
Note that there are multiple tools to derive lower bounds for $p$. Just to mention a few we have Hoeffding's inequality, Chernoff bounds and Clopper-Pearson confidence intervals. The following well known proposition uses the well known Hoeffding's inequality to instantiate $\gamma_m$.
\begin{proposition}
Let $\beta > 0$ and  $\widehat{p}(S_m) = S_m/m$, then with probability at least $1- \beta$: 
\begin{equation*}
    p \geq p(S_m) - \sqrt{\frac{\log(1/\beta)}{2m}}.
\end{equation*}
\end{proposition}

\begin{proposition}
Let $g \colon \calX \to \{0,1\}$ be any function and $\beta > 0$. Let $(x_i, y_i)_{i=1}^m \sim \mathbb{D}$ be an i.i.d sample. Let $S_m = \sum_{i=1}^m \ind_{g(X_i) = Y_i}$. If $\widehat p$ is a $(\gamma_m, \beta)$ lower bound estimator, then with probability at least $1- \beta$ we have
$$H_\epsilon (P || Q) \geq (1 + e^\epsilon ) (\widehat p(S_m) - \gamma_m) - e^\epsilon.$$
\end{proposition}
\begin{proof}
The proof is trivial since $S_m$ is distributed as a binomial random variable for a fixed function $g$. 
\end{proof}
The above proposition shows that, given any classifier $g$, we can find a lower bound on $H_\epsilon$. A natural question is how to find such a classifier. Here we borrow from the extensive work on binary classification and let 
\begin{equation}
\label{eq:logistic}
L(\widehat y, y) := -y \log \widehat y - (1 - y) \log (1 - \widehat y)
\end{equation}
be the logistic loss. The way we choose $g$ is by minimizing $L$ on a training sample. To more precise, we present the subroutine for estimating $H_\epsilon(\cdot\,\|\,\cdot)$, using inputs $P = \calM(D_0)$ and $Q = \calM(D_1)$, in Algorithm~\ref{alg:hs}.

\begin{algorithm}
    \caption{Estimate Hockey Stick divergence}
    \begin{algorithmic}[1]
    \label{alg:hs_divergence}
    \STATE {\bfseries Inputs:} Distributions $P$ and $Q$; $\widehat{p}$ as in Definition~\ref{eq:lb_est}; $\gamma_m > 0$; $\epsilon>0$; and sample size $m$;
    \label{alg:hs}
    \STATE Define a distribution $\mathbb{D}$  as in Lemma~\ref{lem:classification} with input distributions $P$ and $Q$.
    \STATE Generate samples $\{(x_i, y_i)\}_{i=1}^m, (x_i', y_i')_{i=1}^m \sim \mathbb{D}^m$.
    \STATE Compute $h^* = \arg\min_{h \in H} \frac{1}{m} \sum L(h(x_i), y_i)$ where $L$ is as in \eqref{eq:logistic}.
    \STATE Define $g(x) = \ind_{h^*(x) > 1/2}$ for every $x$.
    \RETURN {$(1+e^\epsilon)[\widehat p\left(\sum_{i=1}^m \ind_{g(x'_i) = y_i')}\right) - \gamma_m] - e^\epsilon$.}
\end{algorithmic}
\end{algorithm}

Using the above estimator, we now present how it can be used to implement Algorithm~\ref{alg:general} in Algorithm~\ref{alg:hs_tester}.

\begin{algorithm}[tb]
  \caption{\hstester}
  \label{alg:hs_tester}
    \begin{algorithmic}[1]
    \STATE {\bfseries Inputs:}   \alglinelabel{algln:hsinput} Parameters $\calM$, $T$, $\mathcal G$, and $\beta$ as in Algorithm~\ref{alg:general}; parameters $\widehat{p}$ and $\gamma_m$ as in Algorithm~\ref{alg:hs_divergence}; approximate DP parameters $(\epsilon,\delta)$; and sample size $m$;
    \STATE Define the estimator $\widehat{D}(P\,\|\,Q)$ to be the output of Algorithm~\ref{alg:hs_divergence} with inputs $P$, $Q$, $\widehat{p}$, $\gamma_m$, $\epsilon$, and $m$.
    \STATE Set $\tau = \delta$.
    \STATE{Run the same steps as in Algorithm~\ref{alg:general} with inputs $\calM$, $T$, $\mathcal G$, $\widehat{D}$, $\tau$, and $\beta$.}
    \end{algorithmic}
\end{algorithm}

\subsubsection{MMD-tester}

The definition of approximate differential privacy requires that a mechanism's output distributions for two neighboring datasets are not substantially different from each another. Since the statistical literature describes many distance metrics between distributions, another approach to auditing a differentially private mechanism is to relate the magnitude of one of these distances to the strength of the claimed privacy guarantee, and then apply a statistical test for estimating the distance. An example of such a distance is the \emph{maximum mean discrepancy}
\begin{equation}
\MMD(P, Q, H) \triangleq \sup_{f \in H : \norm{f} \le 1} \left(\int f dP - \int f dQ\right) \label{eq:mmd}
\end{equation}
which is the maximum difference between the expected values with respect to distributions $P$ and $Q$ of any function in the unit ball of a reproducing kernel Hilbert space $H$. Theorem \ref{thm:mmd} gives a relationship between maximum mean discrepancy and approximate differential privacy, and is proved in the Appendix.

\begin{theorem} \label{thm:mmd} Let $H$ be a reproducing kernel Hilbert space with domain $\calX$ and kernel $k$. Suppose $k(x, x) \le 1$ for all $x \in \calX$. If mechanism $\calM : \calD \rightarrow \calX$ satisfies $(\eps, \delta)$-differential privacy then
\[
\MMD(\calM(D_0), \calM(D_1), H) \le e^\eps - 1 + (1 + e^{-\eps})\delta
\]
for any pair of neighboring datasets $D_0, D_1 \in \calD$. 
\end{theorem}

Note that right-hand side in the statement of Theorem \ref{thm:mmd} approaches zero as $\eps \rightarrow 0$ and $\delta \rightarrow 0$. Therefore a positive lower bound on the maximum mean discrepancy between $\calM(D_0)$ and $\calM(D_1)$ serves as a privacy audit of mechanism $M$. One way to estimate the maximum mean discrepancy would be to draw samples from $\calM(D_0)$ and $\calM(D_1)$ and then explicitly perform the optimization over the function class $H$ from Eq.~\eqref{eq:mmd}. Such an approach would resemble the estimation of a lower bound on the hockey stick divergence described in Section \ref{sec:hs_div}. However \cite{gretton2012kernel} showed that no such optimization is necessary, since the maximum mean discrepancy equals the expected value of a simple expression that depends only on the kernel of $H$. Theorem \ref{thm:kerneltest} is a succinct statement of one of the main results of \cite{gretton2012kernel}. 

\begin{theorem}[\cite{gretton2012kernel}] \label{thm:kerneltest} Let $X$ and $X'$ be independent random variables with distribution $P$. Let $Y$ and $Y'$ be independent random variables with distribution $Q$. Let $H$ be a reproducing kernel Hilbert space with domain $\calX$ and kernel $k$. If $k(x, x) \le 1$ for all $x \in \calX$ then
\[
\MMD(P, Q, H)^2 = \E\left[h(X, X', Y, Y')\right]
\]
where $h(x, x', y, y') = k(x, x') - 2k(x, y) + k(y, y')$.
\end{theorem} 

Theorem \ref{thm:kerneltest} shows that it is possible to estimate maximum mean discrepancy by estimating the right-hand side of the theorem statement. Algorithm~\ref{alg:mmd_divergence} and Algorithm~\ref{alg:mmd_test} describe a test of approximate differential privacy based on this estimator, and Theorem \ref{thm:mmdbound} gives a theoretical guarantee of its performance. A proof of Theorem \ref{thm:mmdbound} can be found in the Appendix. While we could have directly used the estimator from \cite{gretton2012kernel}, we modified it to take advantage of the possibility of a low variance sample, using a data-dependent version of Bernstein's inequality due to \cite{maurer2009empirical}.

\begin{algorithm}
  \caption{Estimate MMD divergence}
\begin{algorithmic}[1]
    \label{alg:mmd_divergence}
  \STATE {\bfseries Input:} Distributions $P$ and $Q$; kernel function $k$; $\beta > 0$; $\epsilon > 0$; and sample size $n$; 
  \STATE Generate samples $X, X'\sim P^n$ and $Y, Y' \sim Q^n$.
  \STATE Define $h(x, x', y, y') := k(x, x') - 2k(x, y) + k(y, y')$.
  \STATE Compute 
  \begin{align*}
      \hat{\mu} &= \frac{1}{n}\sum_{i=1}^n h(X_i, X'_i, Y_i, Y'_i) \\
      \hat{s} &= \frac{1}{n}\sum_{i=1}^n h(X_i, X'_i, Y_i, Y'_i)^2
  \end{align*}
  \STATE Set $\hat{\sigma}^2 = \hat{s} - \hat{\mu}^2$.
  \STATE Compute 
  \begin{align*}
  \hat{m} &= \hat{\mu} - \sqrt{\frac{2\hat{\sigma}^2 \log 2/\beta}{n}} - \frac{28 \log 2/\beta}{3(n-1)} - (e^\eps - 1)^2.
  \end{align*}
  \RETURN{the estimate
  \[
  \frac{\sqrt{(e^\eps - e^{-\eps})^2 + (1 + e^{-\eps})\hat{m}} - (e^\eps - e^{-\eps})}{(1 + e^{-\eps})}.
  \]
  }
\end{algorithmic}
\end{algorithm}

\begin{algorithm}
  \caption{\mmdtester}
\begin{algorithmic}[1]
    \label{alg:mmd_test}
  \STATE {\bfseries Input:} Parameters $\calM$, $T$, $\mathcal G$, and $\beta$ as in Algorithm~\ref{alg:general}; parameter $k$ as in Algorithm~\ref{alg:mmd_divergence}; approximate DP parameters $(\epsilon,\delta)$; and sample size $n$;
  \STATE Define the estimator $\widehat{D}(P\,\|\,Q)$ to be the output of Algorithm~\ref{alg:mmd_divergence} with inputs $P$, $Q$, $k$, $\beta$, $\epsilon$ and $n$.
    \STATE Set $\tau = \delta$.
    \STATE{Run the same steps as in Algorithm~\ref{alg:general} with inputs $\calM$, $T$, $\mathcal G$, $\widehat{D}$, $\tau$, and $\beta$.}
\end{algorithmic}
\end{algorithm}


\begin{theorem} \label{thm:mmdbound} If mechanism $\calM$ satisfies $(\eps, \delta)$-differential privacy then Algorithm \ref{alg:mmd_test} does not find a privacy violation with probability at least $1 - \beta$.
\end{theorem}
 
\section{Dataset Search}
\label{sec:dataset-search}

As seen in \cref{sec:building-blocks}, \cref{alg:general}, auditing tests with only black-box access to a mechanism $\calM$ consist of two stages: (1) generation of neighboring datasets $D_0$ and $D_1$, and (2) divergence estimation over $\calM(D_0)$, $\calM(D_1)$. Below we discuss the dataset generation step.

The methods proposed in the literature fall in two categories:
\begin{itemize}
    \item Exploratory methods: These methods ignore any knowledge about the algorithm and/or potential bugs. In this category falls randomized search, that generates datasets without any side information independently at random over trials. A second exploratory method is to divide the search space over a grid and deterministically evaluate the test over the grid. These methods can be useful  when knowledge about the mechanism is available and randomized search can have guarantees of finding a dataset where the guarantee will break. For example, \cite{AKOOMS23} shows that for high-dimensional gaussians, random canaries over the unit sphere are guaranteed to sample a dataset where the privacy parameters are maximized. 
    \item Explotation methods: For black-box testing, the literature has mostly considered crafted datasets that assume certain datasets can be more vulnerable, e.g. where the sensitivity of the mechanism is maximized. While this methods can be very efficient to catch certain bugs, they may fail at finding others that do not happen at the predicted datasets.
\end{itemize}

We propose a new heuristic that balances these two frameworks by leveraging a gaussian process bandit \cite{DKB14} to tradeoff explotation and exploration and find datasets that maximize a specific divergence in \cref{alg:general}. Specifically, assuming that we have access to $$R_\alpha^{h,n}: (D,D') \subseteq \calX \times \calX \to \mathbb{R}.$$ Our goal is to produce a sequence $(D_t,D_t')_t$ that approaches the optimum. 
 In our experiments we will use an open-sourced implementation of the well known Bayesian optimization software Vizier \cite{GSM17}.
\section{Experiments}
\label{sec:experiments}
We evaluate the different testers and dataset finders along several dimensions. We begin by comparing the testers' ability to detect different bugs in non-private randomized mechanisms. Then, we evaluate their false positive rate for flagging private mechanisms as non-private.
We then move to evaluate different dataset search schemes described in section 4 with all testers. Finally, we provide a complexity analysis of different testers run time and sample complexity.

\textbf{Hyperparameters.} Unless stated otherwise, all testers use a probability of failure $\beta=1/3$. \renyitester uses $\alpha=1.5$, approximate DP testers set $\delta=0.01$ for approximate DP tests and $\delta=0$ for pure DP. \gilberttester uses a universe size of 100. For all the experiments we report the average over 10 runs. 

DP-Sniper was run on our mechanisms, however, the tester always failed to complete, even for small sample sizes; thus, it is omitted from our results.

\subsection{Evaluated mechanisms}
\textbf{Pure DP mean mechanisms}. The first three mechanisms attempt to privately compute the mean by generating the random estimates
\begin{gather*}
\dpmean(X) := \frac{\sum_{i=1}^n X_i}{\tilde{n}} + \rho_1, \\
\nondplaplaceone(X) := \frac{\sum_{i=1}^n X_i}{n}+  \rho_2  \\
\nondplaplacetwo(X) := \frac{\sum_{i=1}^n X_i}{n}+ \rho_1
\end{gather*}

 where  $\tilde{n} = \max\{10^{-12}, n + \tau\}$, $\tau\sim{\rm Laplace}(0,2/\varepsilon)$, $\rho_1 \sim{\rm Laplace}(0,2/[\tilde{n}\varepsilon])$, and $\rho_2 \sim {\rm Laplace}(0,2/[n \varepsilon])$. \dpmean satisfies $\epsilon$-DP, \nondplaplaceone one violates the guarantee because it has access to the private number of points, and \nondplaplacetwo one privatizes the number of points to estimate the scale of the noise added to the mean statistic but the mean itself is computed using the non-private number of points. We also evaluate \joint an implementation of a private median mechanism, and \noisymax a private noisy max mechanism for selecting an index with highest counts from a finite set of $[n]$ indices. 
 
\textbf{Sparse vector technique mechanisms}. The next six mechanisms address different private and non-private implementations of the sparse vector technique (SVT), a mechanism for releasing a stream of $c$ queries on a fixed dataset. SVT mechanisms compare each query value against a threshold and the  given algorithm returns certain outputs for a maximum number of queries $c$. We denote these by $\svtone$--$\svtsix$ and they correspond to Algorithms~1-6 in \cite{LSL16}. $\svtone$ and  $\svttwo$  satisfy $\epsilon$-DP. $ \svtfour $ satisfies $(\frac{1+6c}{4})$-DP, and $\svtthree, \svtfive$, and $\svtsix$ do not satisfy $\epsilon$-DP for any finite $\epsilon$.

 \textbf{R\'enyi DP mean mechanisms}. To verify the ability of our tester to detect violations of R\'enyi differential privacy we first instantiate  \nondpgaussianone and \nondpgaussiantwo, analogs of \nondplaplaceone and \nondplaplacetwo but that add Gaussian noise instead of Laplace noise.

\subsection{Bug catching rate}
\textbf{False positives. } Note that from an information theoretic point of view, no tester can claim a truly private mechanism is not private. Therefore, the only reason for a potential false positive comes from a sampling error which is controlled by the parameter $\beta$. In our experiments we verify this is true as none of our testers claim a private mechanism to be non-private.

\textbf{True positives.} Table~\ref{tbl:detection} 
 introduces the true negative rate detection for each tester. For this table, we fixed a pair of neighboring datasets where a violation is known to occur. We ran each tester 10 times for two different privacy regimes, high privacy ($\epsilon=0.01$) and low privacy ($\epsilon=1.0$). Due to discretization, the HistogramTester does not apply to mechanisms whose output has more than one-dimension and consequently does not apply to SVT mechanisms. We observe that on the specific datasets provided, the HistogramTester detects only \nondplaplacetwo independently of the sample size but the fails at detecting other bugs. The \renyitester misses all the SVT mechanisms but catches the rest for the high privacy regime. The MMD tester catches version 1 of non-private mean but fails at catching the rest. Finally, the HockeyStick tester catches all but \svtthree. For SVT, it benefits from the high privacy regime while one-dimensional it varies between these two.

\newcommand{\mc}[3]{\multicolumn{#1}{#2}{#3}}
\definecolor{redish}{rgb}{0.91, 0.39, .38}
\newcommand{\hl}{\cellcolor{redish}}
\begin{table*}[]
\caption{Detection rate of non-private mechanisms. We report the number of times each tester detects a privacy violation on a fixed dataset where a privacy violation is known to occur. Each tester is run 10 times. We report for high privacy regime ($\epsilon=0.01$) and low privacy regime ($\epsilon=1.0$)}
\label{tbl:detection}
\arxiv{\tiny}
\centering
\begin{tabular}{@{}|c|c|ccc|ccc|ccc|ccc|@{}}
\bottomrule
Mechanism                   & $\epsilon$ & \mc{3}{c}{Historgram}& \mc{3}{c}{HS}   & \mc{3}{c}{R\'enyi} &\mc{3}{c}{MMD} \\ \specialrule{.1em}{.05em}{.05em} 
n  (In thousands)           &            &$50$&$100$&$500$   &$50$&$100$&$500$ & $50$ &$100$& $500$ &$50$ &$100$& $500$ \\\specialrule{.1em}{.05em}{.05em} 
\nondplaplaceone              & 0.01       &  \hl  0  & \hl  0 &         \hl  0  & 9& 8 & 9 &     10  & 10  & 10 & 10 & 9 & 10    \\ \cline{2-14}
        & 1.0        &  \hl  0 & \hl  0 &     \hl  0        & 10&10&10        &  \hl0&\hl0&\hl0& \hl  0 &   \hl  0  &    \hl  0 \\ \specialrule{.1em}{.05em}{.05em} 
\nondplaplacetwo                & 0.01        & 10 & 10 &  10      & 10&10&10        &  10     & 7 & 5 & 0& 1&    0  \\ \cline{2-14}
        & 1.0        &  \hl  0   & \hl  0  &  \hl  0       & \hl 0 & \hl 0 & \hl  0         &  \hl  0     & \hl  0&\hl  0    &\hl0 &\hl0 &\hl0\\ \specialrule{.1em}{.05em}{.05em} 
\nondpgaussianone              & 0.01        &\hl 0  & \hl 0& \hl 0      & \hl 0& \hl 0 & \hl 0        &    10&10&10   &  10&4&10   \\ \cline{2-14}
.       & 1.0        &    \hl 0 & \hl 0& \hl 0   & 10&10&10             &      \hl 0& \hl 0& \hl0         &    \hl   0 & \hl0 & \hl0  \\ \specialrule{.1em}{.05em}{.05em} 
\nondpgaussiantwo              & 0.01        &   \hl 0   &\hl  0& \hl 0        & 6 & 5 & 5   & 10   &     10  &  10   &0&0&0\\ \cline{2-14}
        & 1.0        &   \hl 0  &\hl 0 &\hl 0       & \hl 0 &\hl 0 &\hl 0     &\hl 0  & \hl  0     &\hl 0 & \hl0& \hl0& \hl0     \\ \specialrule{.1em}{.05em}{.05em} 
SVT3                        & 0.01        &       & - &         &\hl  0 & \hl 0 & \hl 0       & \hl 0  & \hl 0   & \hl 0     & \hl0  & \hl0 & \hl0   \\ \cline{2-14}
        & 1.0        &       & - &         & \hl 0 & \hl 0 & \hl 0       & \hl 0  & \hl 0   &\hl  0     &  \hl  0 & \hl  0 & \hl  0   \\ \specialrule{.1em}{.05em}{.05em} 
SVT4                        & 0.01        &       & - &         & \hl 0 & \hl 0 & \hl 0       & \hl  0  & \hl 0   & \hl 0     &   1&  2&    0\\ \cline{2-14}
        & 1.0        &       & - &         & 10& 10& 10      & \hl 0  &\hl  0   & \hl 0     &  0 & 0 &  0  \\ \specialrule{.1em}{.05em}{.05em} 
SVT5                        & 0.01        &       & - &         & \hl 0 & \hl 0 & \hl 0       & \hl 0  & \hl 0   &\hl  0     &  \hl 0&\hl 0 &\hl 0   \\ \cline{2-14}
        & 1.0        &       & - &         & 10& 10& 10      & \hl 0  & \hl 0   & \hl 0     &  \hl 0& \hl 0& \hl   0\\ \specialrule{.1em}{.05em}{.05em} 
SVT6                        & 0.01        &       & - &         & \hl 0 & \hl 0 & \hl 0       & \hl 0  & \hl 0   & \hl 0     &  \hl0 &\hl 0 & \hl0   \\ \cline{2-14}
        & 1.0        &       & - &         & 10& 10& 10     & \hl 0  & \hl 0   & \hl 0     &\hl0   &\hl 0 & \hl0   \\ \specialrule{.1em}{.05em}{.05em} 
 \bottomrule
\end{tabular}
\end{table*}

\subsection{Dataset Finder algorithms}
Table~\ref{tbl:dataset-finders} compares three different approaches for finding datasets where privacy violations occur. We observe that all approaches tend to find a dataset for a given bug, but grid search, that pursues an exhaustive deterministic evaluation of the dataset space tends to find datasets at a slower rate than randomized search algorithms.

\begin{table*}[]
\centering
\caption{Average number of trials to find a bug. We run each tester 10 times. For each pair of (Tester, $\calG$) we report the average number of trials the algorithm takes to find a privacy violation, with a maximum of 50 trials. \renyitester uses $\alpha=1.5$. Histogram Tester uses a discretization parameter of $m=100$. All testers use $n=500000$ samples and $\epsilon= 0.01$}
\label{tbl:dataset-finders}
\begin{tabular}{@{}l|cccc@{}}
\toprule
Mechanism (down)                & \multicolumn{4}{c}{Gaussian process bandit} \\ \midrule
Dataset finder                  & Histogram            & Renyi                  & MMD      & HS   \\
\nondplaplaceone             & 50          &     1.0         &      1.0    & 1.2       \\
\nondplaplacetwo          & 8.3 &    9.5        &       12.3   &  2.5     \\
\nondpgaussianone          &   50.0              &  1.0            &     1.0     &  50.0      \\
\nondpgaussiantwo          &   50.0              &  3.6   &       47.3   &  19.1    \\
SVT 3                           &   -                 &     50.0               &        50.0  &   50.0   \\
SVT 4                           &   -                 &     50.0               &        50.0  &   50.0   \\
SVT 5                           &   -                 &  50.0           &         50.0 &   50.0  \\
SVT 6                           &   -                 &    50.0                &         50.0 &   50.0   \\ \midrule
            &  \multicolumn{4}{c}{Randomized}  \\ \midrule
\nondplaplaceone             &  50.0  &     1.0                 &    1.0   &  1.3        \\
\nondplaplacetwo              &   10.1  & 3.6               &  12.3     &   1.5      \\
\nondpgaussianone          &   50.0       &    1.0                  &   1.0    &    50       \\
\nondpgaussiantwo           &   50.0       &  1.2   &     47.3  &  45.6     \\
SVT 3             &       -      &    50.0                 & 50.0      &    50   \\
SVT 4             &    -         &    50.0                 &     50.0  &      50     \\
SVT 5                           &    -         &    50.0             &     50.0  & 50         \\
SVT 6                           &    -         &    50.0                 &   50.0    & 50        \\ \midrule
&  \multicolumn{4}{c}{Grid search} \\ \midrule
\nondplaplaceone               & 50.0 &   1.0      &  1.0  & 1.1     \\
\nondplaplacetwo             &  50.0     &   3.0      &  12.3  & 30.2      \\
\nondpgaussianone              & 50.0      &  1.0        &  1.0  & 50.0      \\
\nondpgaussiantwo              & 50.0      & 11.3   &  47.3  & 50.0      \\
SVT 3                           &     -     &  50.0       & 50.0     & 50.0     \\
SVT 4              &     -     & 50.0         &  50.0   & 50.0      \\
SVT 5                &     -     & 50.0   &   50.0  & 50.0     \\
SVT 6        &     -     &  50.0        &   50.0  & 50.0      \\ \bottomrule
\end{tabular}
\end{table*}

\subsection{Running time}
Another important element of a testing library is the running time for each tester. While the purpose of this testing library is not to serve as a unit- test but potentially as an integration test (allowing for a more prolonged testing time), it remains crucial to ensure that running times are not prohibitive to ensure adoption of this library. In Figure~\ref{fig:time_plots} we show the running time as a function of the sample size of running different tests for a single trial of testing the \nondplaplaceone mechanism. Experiment were ran using a Colab notebook with two different back-ends. The first back-end used only a CPU while the second backend used a NVIDIA V100 GPU. We ran the testers varying the number of samples across $\{50000, 100000, 500000\}$. Each tester was ran 10 times and Figure~\ref{fig:time_plots} shows the average running time across the 10 iterations. As one can see, the hockey stick divergence tester running on CPUs seems to have the longest running time yet when training the discriminator model using GPUs, the run time reduces significantly and in fact is faster than most other testers for sample size 500,000. Note also that,even though both the R\'enyi tester and the hockey-stick divergence tester both train a neural network, their running times are different and the benefits from the hardware accelerators seem to be bigger for the hockey-stick divergence test. This could be because we run the hockey-stick divergence tester for a longer period of time (30 epochs vs 5 epochs) in order to achieve convergence.

\narxiv{
\begin{figure}[t]
    \centering
    \begin{subfigure}{0.45\textwidth}
    \centering
    \includegraphics[width=0.8\textwidth]{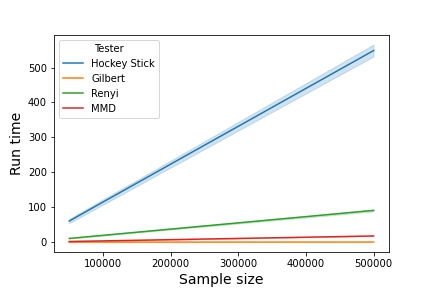} 
    \caption{}
    \end{subfigure}
    \begin{subfigure}{0.45\textwidth}
    \centering
    \includegraphics[width=0.8\textwidth]{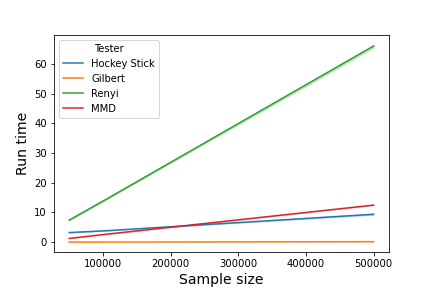}
    \caption{}
    \end{subfigure}
    \caption{Run time (in seconds) for different testers (a) when ran using a single CPU (b) when using a GPU accelerator.}
    \label{fig:time_plots}
\end{figure}
}

\arxiv{
\begin{figure}[t]
    \centering
    \begin{subfigure}{0.49\textwidth}
    \centering
    \includegraphics[width=\textwidth]{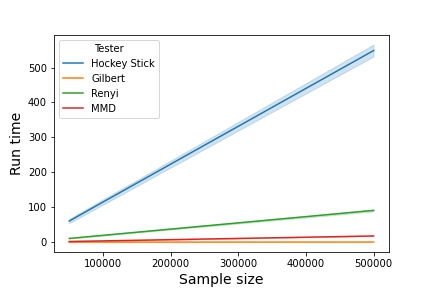} 
    \caption{}
    \end{subfigure}
    \begin{subfigure}{0.49\textwidth}
    \centering
    \includegraphics[width=\textwidth]{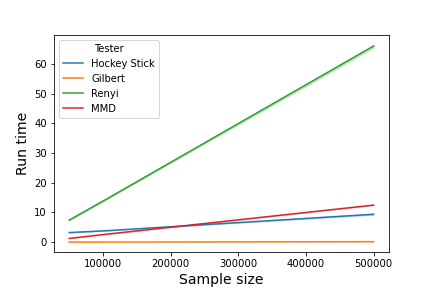}
    \caption{}
    \end{subfigure}
    \caption{Run time (in seconds) for different testers (a) when ran using a single CPU (b) when using a GPU accelerator.}
    \label{fig:time_plots}
\end{figure}
}

\subsection{ScaleGD}
We include a flawed DP-GD's \cite{ACGMMTL16} implementation. The implementation for this mechanism corresponds to a  DP-GD optimizer that receives a model $f_\theta$, learning rate, noise multiplier $\sigma$, clip norm value $G$, and  takes a DP-GD with noise scaled by $\sigma G$ respect to the parameters $\theta$. The implementation simulates a scenario where a developer assumes using a noise multiplier $\sigma_{\text{theory}}$ but in reality uses a noise multiplier $\sigma_{\text{effective}}$. We dub this scenario \scaledSGD. To catch this bug, a one-dimensional model $\theta$ is sufficient so we do not experiment with larger models. Different bugs that could require experimenting with larger models. 
 
We highlight that $\sigma_{\text{effective}} = \frac{\sigma_{\text{theory}}}{2}$ corresponds to a true bug in the literature. It corresponds to a frequent accounting error when using batch or micro-batch clipping instead of per-example clipping in DP-SGD. Per-example clipping is memory and computationally expensive when training high-dimensional models. To address these constraints at the cost of utility, practitioners split a batches of size $n$ into $m$ microbatches of size $n/m$, compute average gradients over each micro-batch, clip and noise the per-microbatch gradient, and finally average the resulting noisy micro-batch gradients. It sometimes goes unnoticed but the sensitivity of per-microbatch gradients is $2G$ instead of $G$. 

We explore this bug on DP-GD using a standard R\'enyi differential privacy accounting framework. The implementation adds gaussian noise with scale $s\cdot \sigma_{\text{theory}}$  where $0<s<1$ is a factor decreasing the assumed noise scale $\sigma_{\text{theory}}$. In \cref{fig:sgd-scale} we show the ability of \renyitester to catch the bug for different values of $s$ given an $(\alpha,\epsilon)-$Renyi DP guarantee. We observe that the divergence estimates follow the true divergence values and the tester catches the bug for $s \leq 0.6$.

\begin{figure}[t]
    \centering
    \includegraphics[scale=0.5]{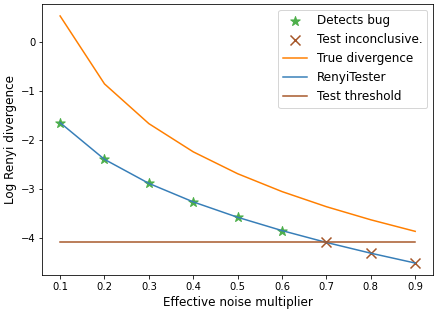}
    \caption{True and estimated value of the R\'enyi divergence on samples of \scaledSGD. }
    \label{fig:sgd-scale}
\end{figure}

\subsection{Effect of sample size}
Although for most of the test runs, the results of Table~\ref{tbl:detection} do not seem to be sensitive to the number of samples drawn from a mechanism. A deeper dive into the divergence evaluation of these mechanisms show that there is in fact a benefit obtained from using larger sample sizes. In particular, for a fixed dataset, Figure~\ref{fig:size_plots}(a) shows that the estimated divergence increases with sample size. Moreover, not surprisingly, the confidence intervals do become smaller as the sample increases. A similar effect can be found on the number of trials required to detect a privacy violation using the Gaussian process bandit algorithm. As one can observe in Figure~\ref{fig:size_plots}(b), the number of trials can drastically decrease --- in view of the reduce variance of divergence estimation --- as the number of samples returned by the mechanism increases.
\begin{figure}[t]
    \centering
        \begin{subfigure}{0.49\textwidth}
    \centering
    \includegraphics[width=0.8\textwidth]{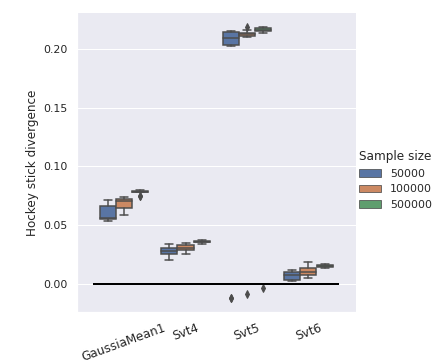}
    \caption{}
    \end{subfigure}
            \begin{subfigure}{0.49\textwidth}
    \centering
    \includegraphics[width=\textwidth]{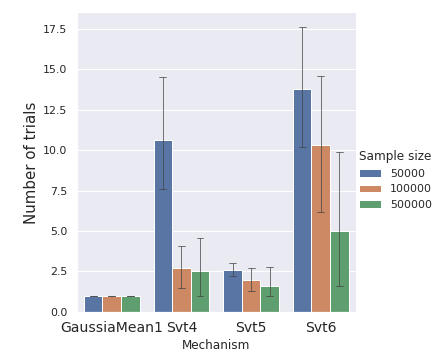}
    \caption{}
    \end{subfigure}
    \caption{Effects of sample size on (a) the divergence estimation of the output mechanism (b) the number of trials needed to find a privacy violation. Privacy parameters in this experiment were $\epsilon = 1$ for the SVT mechanism and $\epsilon = 1.0, \delta = 0.01$ for the Gaussian mechanism}
    \label{fig:size_plots}
\end{figure}

\subsection{Testing anecdotes}
\label{sec:svt4}
It is worth mentioning  that, while preparing the experimental code for this paper, we were under the assumption that the algorithm SVT4 was private. In effect, the mechanism is private albeit with a higher value of $\epsilon$ than what the original introduction of the algorithm claims \cite{LSL16}. We were therefore puzzled by the fact that the Hockey-stick divergence tester would find a privacy violation. It wasn't until further reviewing \cite{LSL16} that we realized that our understanding was wrong and that the tester was indeed correctly flagging a privacy violation on SVT4. 

A similar issue happened when verifying that our tester wouldn't flag private mechanisms as not private. We implemented a mean estimator using a Laplace mechanism to estimate both the numerator and the denominator of the mean. When we ran that algorithm on our tester, the mechanism was flagged as not private. Upon careful inspection of our implementation we realized that we didn't divide the privacy budget to estimate both quantities. This shows how easily bugs can be introduced when implementing private mechanisms.

\section{Conclusion}
We have introduced a new flexible library for testing differential privacy. As demonstrated by even simple implementation issues in our experimentation code, failing to correctly implement a private mechanism can easily happen. These errors have consequences that range from having wrong privacy claims in academic publications (thus making different private algorithms for the same task incomparable) to more serious issues of making public claims around data protection that may not be valid. Having a unified testing framework can help auditors and regulators ensure that private mechanisms are indeed private. It is also a tool that can provide developers with peace of mind in implementing complex private systems. 

We also have introduced a new, powerful way of testing differential privacy through divergence optimization over function spaces. Our results show that this type of function-based estimation consistently outperforms histogram-based divergence estimators. Moreover we provide a variety of these estimators that allow for a better discovery rate of privacy bugs. We hope that by open sourcing this efficient library, the research community can establish a standard of doing end to end testing new differentially private algorithms.  
\bibliographystyle{IEEEtran}
\bibliography{references}

\begin{thebibliography}{10}
\providecommand{\url}[1]{#1}
\csname url@samestyle\endcsname
\providecommand{\newblock}{\relax}
\providecommand{\bibinfo}[2]{#2}
\providecommand{\BIBentrySTDinterwordspacing}{\spaceskip=0pt\relax}
\providecommand{\BIBentryALTinterwordstretchfactor}{4}
\providecommand{\BIBentryALTinterwordspacing}{\spaceskip=\fontdimen2\font plus
\BIBentryALTinterwordstretchfactor\fontdimen3\font minus
  \fontdimen4\font\relax}
\providecommand{\BIBforeignlanguage}[2]{{%
\expandafter\ifx\csname l@#1\endcsname\relax
\typeout{** WARNING: IEEEtran.bst: No hyphenation pattern has been}%
\typeout{** loaded for the language `#1'. Using the pattern for}%
\typeout{** the default language instead.}%
\else
\language=\csname l@#1\endcsname
\fi
#2}}
\providecommand{\BIBdecl}{\relax}
\BIBdecl

\bibitem{GM18}
A.~C. Gilbert and A.~McMillan, ``Property testing for differential privacy,''
  in \emph{Allerton Conference on Communication, Control, and Computing}, 2018.

\bibitem{BSBV21}
B.~Bichsel, S.~Steffen, I.~Bogunovic, and M.~Vechev, ``Dp-sniper: Black-box
  discovery of differential privacy violations using classifiers,'' in
  \emph{Symposium on Security and Privacy (SP)}.\hskip 1em plus 0.5em minus
  0.4em\relax IEEE, 2021.

\bibitem{NWJ10}
X.~Nguyen, M.~J. Wainwright, and M.~I. Jordan, ``Estimating divergence
  functionals and the likelihood ratio by convex risk minimization,''
  \emph{{IEEE} Transactions on Information Theory}, vol.~56, no.~11, 2010.

\bibitem{zhao}
P.~Zhao and L.~Lai, ``Minimax optimal estimation of kl divergence for
  continuous distributions,'' \emph{IEEE Trans. Inf. Theor.}, vol.~66, no.~12,
  p. 7787–7811, dec 2020.

\bibitem{sriperumbudur2012empirical}
B.~K. Sriperumbudur, K.~Fukumizu, A.~Gretton, B.~Sch{\"o}lkopf, and G.~R.
  Lanckriet, ``On the empirical estimation of integral probability metrics,''
  \emph{Electronic Journal of Statistics}, vol.~6, pp. 1550--1599, 2012.

\bibitem{BDKRW21}
J.~Birrell, P.~Dupuis, M.~A. Katsoulakis, L.~Rey-Bellet, and J.~Wang,
  ``Variational representations and neural network estimation of {R}{\'e}nyi
  divergences,'' \emph{SIAM Journal on Mathematics of Data Science}, 2021.

\bibitem{DM22}
C.~Domingo-Enrich and Y.~Mroueh, ``Auditing differential privacy in high
  dimensions with the kernel quantum {R}\'enyi divergence,'' \emph{arXiv
  preprint arXiv:2205.13941}, 2022.

\bibitem{DGKK22}
V.~Doroshenko, B.~Ghazi, P.~Kamath, R.~Kumar, and P.~Manurangsi, ``Connect the
  dots: Tighter discrete approximations of privacy loss distributions,''
  \emph{Proceedings on Privacy Enhancing Technologies}, 2022.

\bibitem{DJT13}
K.~Dixit, M.~Jha, S.~Raskhodnikova, and A.~Thakurta, ``Testing the {L}ipschitz
  property over product distributions with applications to data privacy,'' in
  \emph{Theory of Cryptography Conference (TCC}, 2013.

\bibitem{statdp}
Z.~Ding, Y.~Wang, G.~Wang, D.~Zhang, and D.~Kifer, ``Detecting violations of
  differential privacy,'' in \emph{Proceedings of the 2018 ACM SIGSAC
  Conference on Computer and Communications Security}.\hskip 1em plus 0.5em
  minus 0.4em\relax Association for Computing Machinery, 2018, p. 475–489.

\bibitem{dpsniper}
B.~Bichsel, S.~Steffen, I.~Bogunovic, and M.~Vechev, ``Dp-sniper: Black-box
  discovery of differential privacy violations using classifiers,'' in
  \emph{2021 IEEE Symposium on Security and Privacy (SP)}.\hskip 1em plus 0.5em
  minus 0.4em\relax IEEE, 2021, pp. 391--409.

\bibitem{JUO20}
M.~Jagielski, J.~Ullman, and A.~Oprea, ``Auditing differentially private
  machine learning: How private is private sgd?'' \emph{Advances in Neural
  Information Processing Systems}, 2020.

\bibitem{JE19}
B.~Jayaraman and D.~Evans, ``Evaluating differentially private machine learning
  in practice,'' in \emph{USENIX Security Symposium}, 2019.

\bibitem{ROF21}
S.~Rahimian, T.~Orekondy, and M.~Fritz, ``Differential privacy defenses and
  sampling attacks for membership inference,'' in \emph{ACM Workshop on
  Artificial Intelligence and Security}, 2021.

\bibitem{CYZF20}
D.~Chen, N.~Yu, Y.~Zhang, and M.~Fritz, ``Gan-leaks: A taxonomy of membership
  inference attacks against generative models,'' in \emph{ACM SIGSAC Conference
  on Computer and Communications Security}, 2020.

\bibitem{Guo22}
C.~Guo, B.~Karrer, K.~Chaudhuri, and L.~van~der Maaten, ``Bounding training
  data reconstruction in private (deep) learning,'' in \emph{International
  Conference on Machine Learning {(ICML)}}.\hskip 1em plus 0.5em minus
  0.4em\relax {PMLR}, 2022.

\bibitem{Balle22}
B.~Balle, G.~Cherubin, and J.~Hayes, ``Reconstructing training data with
  informed adversaries,'' in \emph{43rd {IEEE} Symposium on Security and
  Privacy, {SP} 2022, San Francisco, CA, USA, May 22-26, 2022}.\hskip 1em plus
  0.5em minus 0.4em\relax {IEEE}, 2022, pp. 1138--1156.

\bibitem{LMFZ22}
F.~Lu, J.~Munoz, M.~Fuchs, T.~LeBlond, E.~V. Zaresky-Williams, E.~Raff,
  F.~Ferraro, and B.~Testa, ``A general framework for auditing differentially
  private machine learning,'' in \emph{Advances in Neural Information
  Processing Systems}, 2022.

\bibitem{NSTP+}
M.~Nasr, S.~Songi, A.~Thakurta, N.~Papernot, and N.~Carlin, ``Adversary
  instantiation: Lower bounds for differentially private machine learning,'' in
  \emph{2021 IEEE Symposium on security and privacy (SP)}, 2021.

\bibitem{NHSB+23}
M.~Nasr, J.~Hayes, T.~Steinke, B.~Balle, F.~Tram{\`e}r, M.~Jagielski,
  N.~Carlini, and A.~Terzis, ``Tight auditing of differentially private machine
  learning,'' \emph{arXiv preprint arXiv:2302.07956}, 2023.

\bibitem{jagielski2020auditing}
M.~Jagielski, J.~Ullman, and A.~Oprea, ``Auditing differentially private
  machine learning: How private is private sgd?'' \emph{Advances in Neural
  Information Processing Systems}, vol.~33, pp. 22\,205--22\,216, 2020.

\bibitem{AKOOMS23}
G.~Andrew, P.~Kairouz, S.~Oh, A.~Oprea, H.~B. McMahan, and V.~Suriyakumar,
  ``One-shot empirical privacy estimation for federated learning,'' \emph{arXiv
  preprint arXiv:2302.03098}, 2023.

\bibitem{KKPW14}
A.~Krishnamurthy, K.~Kandasamy, B.~Poczos, and L.~Wasserman, ``Nonparametric
  estimation of renyi divergence and friends,'' in \emph{International
  Conference on Machine Learning}.\hskip 1em plus 0.5em minus 0.4em\relax PMLR,
  2014.

\bibitem{M17}
I.~Mironov, ``R{\'e}nyi differential privacy,'' in \emph{IEEE computer security
  foundations symposium (CSF)}, 2017.

\bibitem{ZDW22}
Y.~Zhu, J.~Dong, and Y.-X. Wang, ``Optimal accounting of differential privacy
  via characteristic function,'' in \emph{International Conference on
  Artificial Intelligence and Statistics (AISTATS}, 2022.

\bibitem{BO13}
G.~Barthe and F.~Olmedo, ``Beyond differential privacy: Composition theorems
  and relational logic for f-divergences between probabilistic programs,'' in
  \emph{International Colloquium on Automata, Languages, and
  Programming}.\hskip 1em plus 0.5em minus 0.4em\relax Springer, 2013.

\bibitem{gretton2012kernel}
A.~Gretton, K.~M. Borgwardt, M.~J. Rasch, B.~Sch{\"o}lkopf, and A.~Smola, ``A
  kernel two-sample test,'' \emph{The Journal of Machine Learning Research},
  vol.~13, no.~1, pp. 723--773, 2012.

\bibitem{maurer2009empirical}
A.~Maurer and M.~Pontil, ``Empirical bernstein bounds and sample variance
  penalization,'' 2009.

\bibitem{DKB14}
T.~Desautels, A.~Krause, and J.~W. Burdick, ``Parallelizing
  exploration-exploitation tradeoffs in gaussian process bandit optimization,''
  \emph{Journal of Machine Learning Research}, 2014.

\bibitem{GSM17}
D.~Golovin, B.~Solnik, S.~Moitra, G.~Kochanski, J.~Karro, and D.~Sculley,
  ``Google vizier: A service for black-box optimization,'' in \emph{Proceedings
  of the 23rd ACM SIGKDD international conference on knowledge discovery and
  data mining}, 2017.

\bibitem{LSL16}
M.~Lyu, D.~Su, and N.~Li, ``Understanding the sparse vector technique for
  differential privacy,'' \emph{arXiv preprint arXiv:1603.01699}, 2016.

\bibitem{ACGMMTL16}
M.~Abadi, A.~Chu, I.~Goodfellow, H.~B. McMahan, I.~Mironov, K.~Talwar, and
  L.~Zhang, ``Deep learning with differential privacy,'' in \emph{{ACM SIGSAC}
  conference on computer and communications security}, 2016.

\bibitem{PW22}
Y.~Polyanskiy and Y.~Wu, ``Information theory: From coding to learning,''
  \emph{Book draft}, 2022.

\bibitem{MRT18}
M.~Mohri, A.~Rostamizadeh, and A.~Talwalkar, \emph{Foundations of machine
  learning}.\hskip 1em plus 0.5em minus 0.4em\relax MIT press, 2018.

\end{thebibliography}

\appendix
\section{Additional definitions, lemmas and proofs}
\begin{lemma*}
The variational representation for the Hockey stick divergence (see \cref{eq:var-f-divergence} in the supplementary material) is given by
\begin{equation}
H_{\epsilon}(P||Q) = \sup_{g:0\leq g\leq 1}\Expected{P}{g(X)}-\Expected{Q}{e^{\epsilon}g(X))}.
\end{equation}
\end{lemma*} 

\begin{proof}

By definition, the conjugate function $f^*$ of a function $f$ is given by ${f^*(y) =\sup_{x \in \mathbb{R}} yx - f(x)}$.  We first show that ${f^*(y) = e^{\epsilon}y}$ for $0<y<1$. For the Hockey stick divergence, ${f(x) = \max \{0, x-e^{\epsilon} \}}$, thus, 

\begin{align*}
    f^*(y) &=\sup_{x \in \mathbb{R}} y^Tx - \max \{0, x-e^{\epsilon} \} \\
\end{align*}

If $y<0$, $f^*(y)= \infty$ since $x$ can be arbitrarily negative. 
If $y>1$ and if $x>e^{\epsilon}$ we have that:

\begin{align*}
   f^*(y) =  \sup_{x >e^{\epsilon}} yx - \max \{0, x-e^{\epsilon} \} & = yx - x+e^{\epsilon} \\
    &=(y-1)x +e^{\epsilon}.
\end{align*}
The last expression can be arbitrarily large if $x$ goes to infinity.

Finally, if $0\leq y\leq1$, we consider two cases: $x\leq e^{\epsilon}$ and $x>e^{\epsilon}$ and show the supremum takes the same value in both ranges. When $x\leq e^{\epsilon}$:
\begin{align*}
   f^*(y) =  \sup_{x \leq e^{\epsilon}} yx - \max \{0, x-e^{\epsilon} \} &=\sup_{x \leq e^{\epsilon}} yx = ye^{\epsilon}.\\
\end{align*}
When $x> e^{\epsilon}$,
\begin{align*}
   f^*(y) =  \sup_{x >e^{\epsilon}} yx - \max \{0, x-e^{\epsilon} \} &=\sup_{x >e^{\epsilon}} (y-1)x +e^{\epsilon}.\\
\end{align*}

Since $y\leq1$, $y-1\leq 0$, so $(y-1)x$ is maximized when $x$ takes the smallest value, $x=e^{\epsilon}$. Consequently, 

\begin{align*}
   f^*(y) &=\sup_{x >e^{\epsilon}} (y^T-1)x +e^{\epsilon}.\\
          &= (y-1)e^{\epsilon} +e^{\epsilon} = ye^{\epsilon}.
\end{align*}

Thus, we conclude that $f^*(y)= \begin{cases}
ye^{\epsilon} & \text{ if } 0\leq y\leq1, \\
\infty & \text{ otherwise, }\end{cases}$ and $dom(f*) = [0,1]$.

The result follows by replacing $f^*$ in the the following variational representation of $f-$divergences \cite{PW22}: 
 \begin{equation*}
     D_f(P||Q) = \sup_{g:\calX \to dom(f^*)}\Expected{P}{g(x)} - \Expected{Q}{f^*(g(x))}
 \end{equation*}

\end{proof}



\begin{theorem}[Chernoff multiplicative bound, Theorem D.4. in \cite{MRT18}]
\label{thm:chernoff}
Let $X_1,..., X_n$ be independent random variables drawn according to some distribution $P$ with mean $\mu$ and support in $[0,M]$. Then for any $\eta \in [0,\frac{M}{\mu}-1]$ the following inequalities hold:

\begin{equation}
    \label{eq:chernoff1}
    P\left(\frac{1}{n}\sum_{i=1}^n X_i \geq(1+\eta) \mu\right) \leq e^{\frac{-n\mu\eta^2}{3M}}
\end{equation}

\begin{equation}
    \label{eq:chernoff2}
    P\left(\frac{1}{n}\sum_{i=1}^nX_i \leq(1-\eta) \mu\right) \leq e^{\frac{-n\mu\eta^2}{2M}}
\end{equation}

\end{theorem}

\begin{theorem*}{(Theorem~\ref{thm:sample_complexity_renyi} in main body).}
Let $P$ and $Q$ be two distributions. Let $\Gamma$ be defined as in \cref{def:renyi-partial-defs}, and $h \in \Gamma$, $h\colon  \calX \to \mathbb{R}$ be a function bounded by $C>0$, i.e.,  $\sup_{x \in \calX} |h(x)| < C$. Let $\bX_0 = (X_{0,1},...,X_{0,n})$ and $\bX_1 = (X_{1,1},...,X_{1,n})$ be $n$  realizations of $P$ and $Q$, respectively. Then, if $\eta \in (0,1]$, and $$n \geq \max\left(\frac{3e^{2(\alpha-1)C}\log(2/\beta)}{\eta^2}, \frac{2e^{\alpha C}\log(2/\beta)}{\eta^2} \right),$$ with probability at least  $1-\beta$, it follows that 
\begin{align}
 R_\alpha&(P||Q) \geq
 R^{h,n}_{\alpha}(\bX_0|| \bX_1)  - \log \left(\frac{1+\eta}{1-\eta} \right)
\end{align}
\end{theorem*}

\begin{proof}
Define $\mu_1 = \Expected{P}{e^{(\alpha-1)h(x)}}$, and $\mu_2 = \Expected{Q}{e^{\alpha h(x)}}$, 
and also $M_1 = e^{(\alpha-1)C}$ and $M_2=e^{\alpha C} $. 

From Chernoff's multiplicative bound ( \cref{thm:chernoff}, \cref{eq:chernoff1}) we know that for $\eta_1\in[0, \frac{M_1}{\mu_1}-1]$, with probability less than $e^{\frac{-n\mu_1\eta_1^2}{3M_1}}$, we have
\begin{equation*}
    \frac{\frac{1}{n}\sum_{i=1}^n e^{(\alpha-1)h(x_i)}}{\mathbb{E} [e^{(\alpha-1)h(X)}] } \geq (1+\eta_1)
\end{equation*}
This implies that 
\begin{align*}
    & \log \left[ \frac{1}{n}\sum_{i=1}^n e^{(\alpha-1)h(x_i)}\right] - \log \Expected{P}{e^{(\alpha-1)h(X)}} \\
    &\geq\log(1+\eta_1) \geq \tfrac{\alpha-1}{\alpha }\log(1+\eta_1),
\end{align*}

or equivalently,

\begin{align}
&\tfrac{\alpha}{\alpha-1 }\log \frac{1}{n} \sum_{i=1}^n e^{(\alpha-1)h(x_i)} 
 -\tfrac{\alpha}{\alpha-1}\log \Expected{P} {e^{(\alpha-1)h(X)}} \nonumber  \\
&\geq \log(1+\eta_1) \label{eq:part_one}
\end{align}
Note that by setting $n\geq \frac{3e^{(\alpha-1)C}\log(2/\beta)}{\mu_1 \eta_1^2}$ the above bound holds with probability  at most $\beta/2$. A similar analysis (using \cref{eq:chernoff2}) shows that for $ n\geq \frac{2e^{\alpha C}\log(2/\beta)}{\mu_2 \eta_2^2}$ with probability at most $\beta/2$ the following bound holds:
\begin{align}
\label{eq:part_two}
    &\log \Expected{Q}{e^{\alpha h(X_1)}}
     -\log \left[\frac{1}{n}\sum_{i=1}^n e^{\alpha h(X_{1,i})}\right] \nonumber \\
     &\geq \log\left(\frac{1}{1-\eta_2}\right).
\end{align}
Finally, summing \cref{eq:part_one} and \cref{eq:part_two}, and using the union bound, with probability $1-\beta$ we have that
\begin{align}
\nonumber R^{h,n}(\bX_0||\bX_1) - \log\left(\frac{1}{1-\eta_2}\right) - \log(1+\eta_1) 
\leq R^h(P || Q).
\end{align}
Letting $\eta = \min(\eta_1, \eta_2)$, 

\begin{equation}
\label{eq:tmp1-renyi-comp}
    R^{h,n}(\bX_0||\bX_1) - \log\left(\frac{1+\eta}{1-\eta}\right) 
\leq R^h(P || Q).
\end{equation}

Finally, since $h \in \Gamma$, using the variational formulation for the R\'enyi divergence,
\begin{align}
     \label{eq:tmp2-renyi-comp}
   R_\alpha(P || Q) &= \sup_{h \in \Gamma} R_\alpha^h(P ||Q) \\
         &\geq R_\alpha^h(P || Q)\\
\end{align}

and the result follows by combining  \cref{eq:tmp1-renyi-comp} and \cref{eq:tmp2-renyi-comp}.
\end{proof}

\begin{definition}
\label{def:HS}
Given a function $f:[0,\infty) \to \mathbb{R}$ such that $f(1)=0$ and $\lim_{t\to0^+}f(t) = f(0)$. Let $P$ and $Q$ be two probability measures over space$\calX$ such that $P$ is absolutely continuous respect to $Q$. The $f-$divergence between $P$ and $Q$ is defined as 

\begin{equation}
    D_f(P||Q) = \Expected{Q}{f\left(\frac{dP}{dQ} \right)}
\end{equation}

It will be useful to see $D_f$ through its variational representation:

\begin{equation}
    \label{eq:var-f-divergence}
    D_f(P||Q) = \sup_{g:\mathcal{X} \to \text{dom}(f_{ext}^*)} \E_P\left[g(X)\right]-\E_Q\left[f^*(g(X))\right].
\end{equation}
\end{definition}
\begin{lemma}
 Let $\epsilon > 0$ Fix two distributions $P$ and $Q$ over $\calX$. Let $\mathbb{D}$ be a mixture distribution over $\calX \times \{0, 1\}$ defined as $\mathbb{D} = \frac{e^\epsilon}{1 + e^\epsilon} Q \times \delta_0 + \frac{1}{1 + e^\epsilon} P \times \delta_1$. That is, $(X, Y) \sim \mathbb{D}$ labels samples from $Q$ as $0$ and samples from $P$ as $1$. Then
 $$H_\epsilon(P || Q) = (1 + e^\epsilon) \sup_{g} \mathbb{D}(g(X) = Y) - e^\epsilon$$
 \end{lemma}
\begin{proof}
We begin by showing that the variational formula for the hockey-stick divergence is maximized for functions taking values in $\{0, 1\}$. Using Hahn's decomposition theorem let $A$ and $B$ be such that $\calX = A \cup B$, $A\cap B = \emptyset$ and $A$ is the maximal positive set of the signed measure $\nu = P - e^\eps Q$. That is $P(A) - e^\eps Q(A) > P(A \cap C) - e^\eps Q(A \cap C) > 0$ for every $C \subset \calX$ and $P(B) - e^ep Q(B) \leq (B \cap C) - e^\eps P(B \cap C) < 0$ for every $C$. 
Let $g^*$ be defined as $g^*(x) = 1$ if $x \in A$ and $0$ otherwise. We show that $g^*$ maximizes:
$$\E_P[g(x)] - \E_Q(e^\eps g(X)]$$
First notice that the above expression evaluated on $g^*$ is given by
$$P(A) - e^\eps Q(A).$$ Moreover, by Proposition~\ref{prop:signedmeasure} we also know that for any other function $g$ we must have
$$\E_P[g(X)] - e^\eps\E_Q[g(X)] = \E_\nu[g(X)] \leq \nu(A) = P(A) - e^\eps Q(A)$$
Now notice that by definition of $\mathbb{D}$ we have 
\begin{align*}
&\mathbb{D}(g(X) = Y)\\
&= \frac{e^\eps}{1 + e^\eps} Q(g(X)= 0) + \frac{1}{1 + e^\eps}P(g(X) = 1)\\
&= 
\frac{e^\eps}{1 + e^\eps}( 1 -  Q(g(X)= 1)) +   \frac{1}{1 + e^\eps}(P(g(X) = 1))  \\
& = \frac{1}{1 + e^\eps} \left(P(g(X) = 1) -e^\eps Q(g(X) = 1)\right) + \frac{e^\eps}{1 + e^\eps} \\
&= \frac{1}{1 + e^\eps} \left(\E_P[g(X)] - e^\eps \E_Q[g(X)]\right)+\frac{e^\eps}{1 + e^\eps} \\
\end{align*}
Taking the supremum over $g \colon \calX \to \{0, 1\}$ we obtain the following:

$$\sup_{g} \mathbb{D}(g(X) = Y) - \frac{e^\epsilon}{1 + e^\epsilon} = \frac{1}{1 + e^\epsilon} H_\epsilon(P||Q)$$
\end{proof}

\begin{proposition}
\label{prop:signedmeasure}
Let $g\colon \calX \to [0, 1]$ be a measurable function. Let $\nu$ be a signed, finite measure and let $\E_\nu[g(X)] = \int g(x) d\nu(x)$ denote the Lebesgue integral of $g$. Let $A$ and $B$ be the maximal positive and negative sets of $\nu$ according to Hahn's decomposition theorem then the following bound holds for any $g$:
\begin{equation}
    \E_{\nu}[g(X)] \leq \nu(A)
\end{equation}
\end{proposition}
\begin{proof}
We begin the proof by considering the simple case where $g(x) = \sum_{i=1}^n \alpha_i \ind_{C_i}(x)$ for some disjoint measurable sets $C_i$ satisfying $\bigcup_{i=1}^n C_i = \calX$ and $\alpha_i \in [0, 1]$. We refer to such functions as \emph{simple} functions. In this case
\begin{align*}
    \E_{\nu}[g(X)] &= 
    E_{\nu}[g(X) \ind_A ] + [g(X) \ind_B] \\
    & = \sum_{i=1}^n  \alpha_i \E_{\nu}[\ind_{A\cap C_i} ] + \sum_{i=1}^n\alpha \E_{\nu}[\ind_{B \cap C_i}]\\
    &= \sum_{i=1}^n\alpha_i \nu (A \cap C_i) + \sum_{i=1}^n\alpha_i \nu(B \cap C_i) \\
    &\leq \sum_{i=1}^n \alpha_i \nu(A \cap C_i) 
\end{align*}
Where we have used the  the property of Hahn's decomposition theorem: $ \nu(B \cap C_i) \leq 0$ for all $i$ as well as the fact that $\alpha_i > 0$. Using the same property for the positive set $A$ we have that $\nu(A \cap C_i) > 0$ for all $i$. Moreover since $\alpha_i \leq 1$ we have that the above quantity is bounded by:
\begin{equation*}
    \sum_{i=1}^n \nu(A \cap C_i) = \nu(A)
\end{equation*}
For a general measurable function $g \colon \calX \to [0, 1]$ we use the well known property that one can approximate $g$ as a limit of simple functions $g_n$. By the dominated convergence theorem we then have that 
\begin{equation*}
    \E_\nu[g(x)]  = \E_\nu[\lim_{n\to \infty} g_n(x)] \leq \nu(A).
\end{equation*}
\end{proof}

\begin{proof}[Proof of Theorem \ref{thm:mmd}] Let $k_x = k(x, \cdot)$ for all $x \in \calX$. By the reproducing property we have $k_x \in H$ and $f(x) = \inner{f, k_x}$ for all $f \in H$. In particular this implies $k(x, y) = \inner{k_x, k_y}$. And since $k(x, x) \le 1$ we have
$\norm{k_x} = \sqrt{\inner{k_x, k_x}} = \sqrt{k(x, x)} \le 1$
Therefore
\begin{equation}
|f(x)| = |\inner{f, k_x}| \le \norm{f} \norm{k_x} \le \norm{f}. \label{eq:mmdbounded}
\end{equation}
Let $P = \calM(D_0)$ and $Q = \calM(D_1)$. For any $f \in H$ such that $\norm{f} \le 1$ we have
\begin{align*}
\int f dP - \int f dQ &= \int f \frac{dP}{dQ} dQ - \int f dQ \\
&= \int \left(\frac{dP}{dQ} - 1\right) f dQ\\
&\le \sup_x |f(x)| \int \left|\frac{dP}{dQ} - 1\right| dQ \\
&\le \int \left|\frac{dP}{dQ} - 1\right| dQ
\end{align*}
where $\frac{dP}{dQ}$ is the Radon-Nikodym derivative and the second inequality follows from Eq.~\eqref{eq:mmdbounded}. Now partition $\calX$ as follows:
\begin{align*}
A &= \left\{x \in \calX : \frac{dP}{dQ}(x) \ge 1\right\}\\
B &= \left\{x \in \calX : \frac{dP}{dQ}(x) < 1\right\}
\end{align*}
Continuing from above we have
\[
\int \left|\frac{dP}{dQ} - 1\right| dQ = \int_A \left(\frac{dP}{dQ} - 1\right) dQ + \int_B \left(1 - \frac{dP}{dQ}\right) dQ
\]
We will bound each integral on the right-hand side. We have
\begin{align*}
\int_A \left(\frac{dP}{dQ} - 1\right) dQ &= P(A) - Q(A)\\
&\le e^{\eps} Q(A) + \delta - Q(A)\\
&= Q(A)(e^\eps - 1) + \delta
\end{align*}
where the inequality follows from the first part of the definition of $(\eps, \delta)$-differentially privacy. We also have
\begin{align*}
\int_B \left(1 - \frac{dP}{dQ}\right) dQ &= Q(B) - P(B)\\
&\le Q(B) - (e^{-\eps} Q(B) - e^{-\eps}\delta)\\
&= Q(B)(1 - e^{-\eps}) + e^{-\eps}\delta
\end{align*}
where the inequality follows from the second part of the definition of $(\eps, \delta)$-differentially privacy. Putting it all together, we have
\begin{align*}
&\int f dP - \int f dQ\\
&\le Q(A)(e^\eps - 1) + \delta + Q(B)(1 - e^{-\eps}) + e^{-\eps}\delta\\
&\le Q(A)(e^\eps - 1) + \delta + Q(B)(e^\eps - 1) + e^{-\eps}\delta\\
&= e^\eps - 1 + (1 + e^{-\eps})\delta
\end{align*}
where we also used $e^\eps - 1 \ge 1 - e^{-\eps}$, which can be proved by expanding and rearranging the inequality $(e^\eps - 1)^2 \ge 0$, and $Q(A) + Q(B) = 1$.
\end{proof}

\begin{proof}[Proof of Theorem \ref{thm:mmdbound}] Suppose $P$ and $Q$ satisfy $(\eps, \delta)$-differential privacy. By squaring both sides of Theorem \ref{thm:mmd} we have
\[
\MMD(P, Q, H)^2 \le (e^\eps - 1)^2  + 2(e^\eps - e^{-\eps})\delta + (1 + e^{-\eps})\delta^2.
\]
Let $X$ and $X'$ be independent random variables with distribution $P$. Let $Y$ and $Y'$ be independent random variables with distribution $Q$. By Theorem \ref{thm:kerneltest} we have
\[
\MMD(P, Q, H)^2 = \E[h(X, X', Y, Y')].
\]
By the empirical Bernstein bound \cite{maurer2009empirical} we have with probability $1 - \beta$
\[
\hat{\mu} - \E[h(X, X', Y, Y')] \le \sqrt{\frac{2\hat{\sigma}^2 \log \frac2\beta}{n}} + \frac{28 \log \frac2\beta}{3(n-1)}
\]
where we again used $-2 \le h(X, X', Y, Y') \le 2$. Putting the above together we have
\begin{align*}
0 &\le -\hat{m} + 2(e^\eps - e^{-\eps})\delta + (1 + e^{-\eps})\delta^2.
\end{align*}
Minimizing the right-hand side using the quadratic formula proves the theorem.
\end{proof}






\end{document}